\documentclass[a4paper,reqno]{amsart}
\usepackage[all]{xy} 
\usepackage{amssymb} 
\usepackage{hyperref}
\usepackage{eucal}

\numberwithin{equation}{section}

 1  
\newfont{\msi}{msbm8 scaled \magstephalf}  
\newfont{\msii}{msbm6 scaled \magstephalf}  

\newtheorem{definition}{Definition}[section]

\newtheorem{theorem}[definition]{Theorem}
\newtheorem{proposition}[definition]{Proposition}
\newtheorem{corollary}[definition]{Corollary}
\newtheorem{remarkth}[definition]{Remark}
\newtheorem{example}[definition]{Example}
\newenvironment{remark}{\begin{remarkth}\upshape}{\hfill$\diamond$\end{remarkth}}

\renewcommand{\emph}[1]{{\bfseries\itshape{#1}}}
\newcommand{\comp}{\makebox[7pt]{\raisebox{1.5pt}{\tiny $\circ$}}}

\def\a{\mathrm{ a}}

\newcommand{\R}{\mathbb{R}} 

\newcommand{\lcf}{\lbrack\! \lbrack}
\newcommand{\rcf}{\rbrack\! \rbrack}

\makeatletter
\newcommand\prol{\@ifstar{\@proldf}{\@prolpf}} 
\def\@prolpf{\@ifnextchar[{\@prolpf@wrt}{\@prolpf@}}
\def\@prolpf@wrt[#1]#2{\@ifnextchar[{\@prolpf@wrt@at{#1}{#2}}{\@prolpf@wrt@{#1}{#2}}}

\def\@prolpf@wrt@at#1#2[#3]{\prolsymbol^{#1}_{#3}#2}
\def\@prolpf@wrt@#1#2{\prolsymbol^{#1}#2}
\def\@prolpf@#1{\@ifnextchar[{\@prolpf@at{#1}}{\@prolpf@@{#1}}}
\def\@prolpf@at#1[#2]{\prolsymbol_{#2}#1}
\def\@prolpf@@#1{\prolsymbol#1}
\def\@proldf{\@ifnextchar[{\@proldf@wrt}{\@proldf@}}
\def\@proldf@wrt[#1]#2{\@ifnextchar[{\@proldf@wrt@at{#1}{#2}}{\@proldf@wrt@{#1}{#2}}}

\def\@proldf@wrt@at#1#2[#3]{\prolsymbol^{*#1}_{#3}#2}
\def\@proldf@wrt@#1#2{\prolsymbol^{*#1}#2}
\def\@proldf@#1{\@ifnextchar[{\@proldf@at{#1}}{\@proldf@@{#1}}}
\def\@proldf@at#1[#2]{\prolsymbol^*_{#2}#1}
\def\@proldf@@#1{\prolsymbol^*#1}
\def\prolsymbol{\mathcal{L}}
\makeatother

\newcommand{\X}{\mathcal{X}}
\newcommand{\Y}{\mathcal{Y}}
\newcommand{\V}{\mathcal{V}}
\newcommand{\T}{\mathcal{T}}
\renewcommand{\P}{\mathcal{P}}

\newcommand{\A}{\mathcal{A}}
\newcommand{\M}{\mathcal{M}}
\def\B{{\mathcal B}}
\def\bidual{{\widetilde{\A}}}
\def\dual{{\A^+}}

\newcommand{\J}{\mathcal{J}}

\def\lcf{\lbrack\! \lbrack}
\def\rcf{\rbrack\! \rbrack}
\newcommand{\U}{{\mathcal{U}}}

\newcommand{\pr}{\mathrm{pr}}
\begin{document}

\title[Constrained mechanics on Lie affgebroids]
{Variational constrained mechanics on Lie affgebroids}

\author[J.\ C.\ Marrero]{J.\ C.\ Marrero}
\address{J.\ C.\ Marrero:
Departamento de Matem\'atica Fundamental y Unidad Asociada
ULL-CSIC Geometr\'{\i}a Diferencial y Mec\'anica Geom\'etrica, Facultad
de Ma\-te\-m\'a\-ti\-cas, Universidad de la Laguna, La Laguna,
Tenerife, Canary Islands, Spain} \email{jcmarrer@ull.es}

\author[D.\ Mart\'{\i}n de Diego]{D. Mart\'{\i}n de Diego}
\address{D.\ Mart\'{\i}n de Diego:
Unidad Asociada ULL-CSIC Geometr\'{\i}a Diferencial y Mec\'anica
Geom\'etrica, Instituto de Ciencias Matem\'aticas
(CSIC-UAM-UC3M-UCM), Serrano 123, 28006 Madrid, Spain}
\email{d.martin@imaff.cfmac.csic.es}

\author[D.\ Sosa]{D.\ Sosa}
\address{D.\ Sosa:
Departamento de Econom\'{\i}a Aplicada y Unidad Asociada ULL-CSIC
Geo\-me\-tr\'{\i}a Diferencial y Mec\'anica Geom\'etrica, Facultad de
CC. EE. y Empresariales, Universidad de La Laguna, La Laguna,
Tenerife, Canary Islands, Spain and Department of Mathematics,
Purdue University, 150 N. University Street, West Lafayette,
Indiana 47907-2067, USA} \email{dnsosa@ull.es and
dsosa@math.purdue.edu}



\begin{abstract}
In this paper we discuss variational constrained mechanics
(vakonomic mechanics) on Lie affgebroids. We obtain the dynamical
equations and the aff-Poisson bracket associated with a vakonomic
system on a Lie affgebroid ${\mathcal A}$. We devote special
attention to the particular case when the nonholonomic constraints
are given by an affine subbundle of ${\mathcal A}$ and we discuss
the variational character of the theory. Finally, we apply the
results obtained to several examples.

\end{abstract}

\maketitle

\section{Introduction}
Lie algebroids have deserved a lot of interest in recent years
(from a theoretical and applied point of view). In the context of
Mechanics, an ambitious program was proposed by Weinstein
\cite{We} in order to develop geometric formulations of the
dynamical behavior of Lagrangian and Hamiltonian systems on Lie
algebroids. In the last years, this program has been actively
developed by many authors. In fact, a Klein's formalism for
unconstrained Lagrangian systems on Lie algebroids has been
discussed and a symplectic formulation of Hamiltonian mechanics on
these objects has been developed (see \cite{LMM,M1,M2,PoPo}). The
main notion is that of prolongation of a Lie algebroid over a map
introduced by Higgins and Mackenzie \cite{HM}. An alternative
approach, using the linear Poisson structure on the dual bundle of
a Lie algebroid, was discussed in \cite{GrGrUr3}.

An interesting kind of mechanical systems are those subject to
external linear constraints. For these systems, one may derive the
dynamical equations using the Lagrange-D'Alembert principle
(nonholonomic mechanics) or using a constrained variational
principle (vakonomic mechanics). The resultant equations are, in
ge\-ne\-ral, different. Constrained Lagrangian systems
(variational or not) have application in many different areas:
engineering, optimal control theory, mathematical economics,
sub-Riemannian geometry, motion of microorganisms, etc. For a
geometrical treatment of standard mechanical systems subject to
external linear constraints we remit to the monographs
\cite{Bl,Co} and references therein.

More recently, several authors discuss the more general class of
nonholonomic Lagrangian (Hamiltonian) systems subject to linear
constraints on Lie algebroids (see \cite{CoLeMaMa,CoMa,Me,MeLa}).
In the same Lie algebroid setting, other authors \cite{IMMS2}
consider variational constrained mechanical systems. In another
direction, a unified approach of nonholonomic and vakonomic
mechanics, using general algebroids instead of just Lie
algebroids, was developed in \cite{GrGr}.

As a consequence of all these investigations, one deduces that
there are several reasons for discussing unconstrained
(constrained) Mechanics on Lie algebroids:

i) The inclusive nature of the Lie algebroid framework. In fact,
under the same umbrella, one can consider standard unconstrained
(constrained) mechanical systems, (nonholonomic and vakonomic)
Lagrangian systems on Lie algebras, unconstrained (constrained)
systems evolving on semidirect products or (nonholonomic and
vakonomic) Lagrangian systems with symmetries.

ii) The reduction of a (nonholonomic or vakonomic) mechanical
system on a Lie algebroid is a (nonholonomic or vakonomic)
mechanical system on a Lie algebroid. However, the reduction of an
standard unconstrained (constrained) system on the tangent
(cotangent) bundle of the configuration manifold is not, in
general, an standard unconstrained (constrained) system.

iii) The theory of Lie algebroids gives a natural interpretation
of the use of quasi-coordinates (velocities) in Mechanics
(particularly, in nonholonomic and vakonomic mechanics).

On the other hand, in \cite{GGrU,MMeS} an affine version of the
notion of a Lie algebroid structure was introduced. The resultant
geometric object is called a Lie affgebroid structure. A Lie
affgebroid structure on an affine bundle ${\mathcal A}$ is
equivalent to a Lie algebroid structure on the bidual bundle to
${\mathcal A}$ such that the section of the affine dual to
${\mathcal A}$ induced by the constant map on ${\mathcal A}$ equal
to $1$ is a $1$-cocycle.

Lie affgebroid structures may be used to develop a time-dependent
version of unconstrained Lagrangian and Hamiltonian equations on
Lie algebroids (see \cite{GGU2,IMPS,M,MMeS,SaMeMa}). In addition,
in \cite{IMMS} the authors present a geometric description of
Lagrangian and Hamiltonian systems on Lie affgebroids subject to
affine nonholonomic constraints. If we apply this general theory
to the particular case when the Lie affgebroid is the $1$-jet
bundle of local sections of a fibration $\tau: Q \to \R$ then one
recovers some results obtained in \cite{CLMM,LeMaMa1,LeMaMa2} for
standard time-dependent nonholonomic Lagrangian systems subject to
affine constraints. The same reasons for discussing unconstrained
(constrained) mechanics on Lie algebroids are also valid for
discussing unconstrained (constrained) mechanics on Lie
affgebroids.

On the other hand, in \cite{ViB} the authors discuss standard
time-dependent vakonomic dynamics and its relation with
presymplectic geometry. More recently, in \cite{BLEE} a geometric
approach to time-dependent optimal control problems is proposed.
This formulation is based on the Skinner and Rusk formalism for
Lagrangian and Hamiltonian systems. Some applications are also
presented. The aim of this paper is to extend these formulations
to the Lie affgebroid setting or, in other words, to discuss
vakonomic mechanics on Lie affgebroids.

The paper is organized as follows. In Section 2, we recall some
well-known facts about the geometry of Lie affgebroids and about
the unconstrained Hamiltonian formalism on Lie affgebroids. In
Section 3, we obtain the vakonomic equations and the vakonomic
bracket for a constrained mechanical system on a Lie affgebroid
${\mathcal A}$. We devote special attention to the particular case
when the constraints are given by an affine subbundle of
${\mathcal A}$. We also discuss, in this section, the variational
character of the theory. In section 4, we apply the results
obtained in the paper to several examples. In fact, we develop a
Skinner-Rusk formalism on Lie affgebroids. We also consider
vakonomic Mechanics on a Lie affgebroid ${\mathcal A}$, for the
particular case when ${\mathcal A}$ is the $1$-jet bundle of local
sections of a fibration over $\mathbb{R}$. As a consequence, we
recover some previous results in the literature. We also discuss
optimal control systems as vakonomic systems on Lie affgebroids.
The paper ends with our conclusions and a description of future
research directions.

\newpage

\section{Hamiltonian formalism on Lie affgebroids} \label{sec2}

\subsection{Lie affgebroids}\label{secaff}

Let $\tau_{\mathcal A}:{\mathcal A}\rightarrow Q$ be an affine
bundle with associated vector bundle $\tau_V:V\rightarrow Q$.
Denote by $\tau_{{\mathcal A}^+}:{\mathcal A}^+={{\mathcal
A}\hspace{-0.05cm}f\hspace{-0.075cm}f}({\mathcal A},\R)\rightarrow
Q$ the dual bundle whose fibre over $x\in Q$ consists of affine
functions on the fibre ${\mathcal A}_x$. Note that this bundle has
a distinguished section $1_{\mathcal A}\in\Gamma(\tau_{{\mathcal
A}^+})$ corresponding to the constant function $1$ on ${\mathcal
A}$. We also consider the bidual bundle
$\tau_{\widetilde{{\mathcal A}}}:\widetilde{{\mathcal
A}}\rightarrow Q$ whose fibre at $x\in Q$ is the vector space
$\widetilde{{\mathcal A}}_x=({\mathcal A}_x^+)^*$. Then,
${\mathcal A}$ may be identified with an affine subbundle of
$\widetilde{{\mathcal A}}$ via the inclusion $i_{\mathcal
A}:{\mathcal A}\rightarrow\widetilde{{\mathcal A}}$ given by
$i_{\mathcal A}(\a)(\varphi)=\varphi(\a)$, which is an injective
affine map whose associated linear map is denoted by
$i_V:V\rightarrow\widetilde{{\mathcal A}}$. Thus, $V$ may be
identified with a vector subbundle of $\widetilde{{\mathcal A}}$.

A {\it
 Lie affgebroid structure} on ${\mathcal A}$ consists of a Lie algebra structure
 $\lcf\cdot,\cdot\rcf_V$ on the space
 $\Gamma(\tau_V)$ of the sections of $\tau_V:V\rightarrow Q$, a $\R$-linear action
 $D:\Gamma(\tau_{\mathcal A})\times\Gamma(\tau_V)\rightarrow\Gamma(\tau_V)$ of
 the sections of ${\mathcal A}$ on $\Gamma(\tau_V)$ and an affine map
 $\rho_{\mathcal A}:{\mathcal A}\rightarrow TQ$, the {\it anchor map}, satisfying the following
 conditions:
  \begin{enumerate}
\item[$ \bullet$] $D_X\lcf\bar{Y},\bar{Z}\rcf_V=\lcf
D_X\bar{Y},\bar{Z}\rcf_V+\lcf\bar{Y},D_X\bar{Z}\rcf_V,$
\item[$\bullet$]
$D_{X+\bar{Y}}\bar{Z}=D_X\bar{Z}+\lcf\bar{Y},\bar{Z}\rcf_V,$
\item[$\bullet$] $D_X(f\bar{Y})=fD_X\bar{Y}+\rho_{\mathcal
A}(X)(f)\bar{Y},$
\end{enumerate}
\noindent for $X\in\Gamma(\tau_{\mathcal A})$,
$\bar{Y},\bar{Z}\in\Gamma(\tau_V)$ and $f\in C^{\infty}(Q)$ (see
\cite{GGrU,MMeS}).

If $(\lcf\cdot,\cdot\rcf_V,D,\rho_{\mathcal A})$ is a Lie
affgebroid structure on an affine bundle ${\mathcal A}$ then $(V,$
$\lcf\cdot,\cdot\rcf_V,\rho_V)$ is a Lie algebroid, where
$\rho_V:V\rightarrow TQ$ is the vector bundle map associated with
the affine morphism $\rho_{\mathcal A}:{\mathcal A}\rightarrow TQ$
(for the definition and properties of Lie algebroids we remit to
\cite{Ma}).

A Lie affgebroid structure on an affine bundle $\tau_{\mathcal
A}:{\mathcal A}\rightarrow Q$ induces a Lie algebroid structure
$(\lcf\cdot,\cdot\rcf_{\widetilde{{\mathcal
A}}},\rho_{\widetilde{{\mathcal A}}})$ on the bidual bundle
$\widetilde{{\mathcal A}}$ such that $1_{\mathcal
A}\in\Gamma(\tau_{{\mathcal A}^+})$ is a $1$-cocycle in the
corresponding Lie algebroid cohomology, that is,
$d^{\widetilde{{\mathcal A}}}1_{\mathcal A}=0$. Here,
$d^{\widetilde{{\mathcal A}}}$ is the differential of the Lie
algebroid $(\widetilde{\mathcal A}, \lcf \cdot , \cdot
\rcf_{\widetilde{\mathcal A}}, \rho_{\widetilde{{\mathcal A}}})$.

Conversely, let $(U,\lcf\cdot,\cdot\rcf_U,\rho_U)$ be a Lie
algebroid over $Q$ and $\phi:U\rightarrow\R$ be a $1$-cocycle of
$(U,\lcf\cdot,\cdot\rcf_U,\rho_U)$ such that $\phi_{|U_x}\neq 0$,
for all $x\in Q$. Then, ${\mathcal A}=\phi^{-1}\{1\}$  is an
affine bundle over $Q$ which admits a Lie affgebroid structure in
such a way that $(U,\lcf\cdot,\cdot\rcf_U,\rho_U)$ may be
identified with the bidual Lie algebroid $(\widetilde{{\mathcal
A}},\lcf\cdot,\cdot\rcf_{\widetilde{{\mathcal
A}}},\rho_{\widetilde{{\mathcal A}}})$ to ${\mathcal A}$ and,
under this identification, the $1$-cocycle $1_{\mathcal
A}:\widetilde{{\mathcal A}}\rightarrow\R$ is just $\phi$. The
affine bundle $\tau_{\mathcal A}:{\mathcal A}\rightarrow Q$ is
modelled on the vector bundle $\tau_V:V=\phi^{-1}\{0\}\rightarrow
Q$.

Let $\tau_{\mathcal A}:{\mathcal A}\to Q$ be a Lie affgebroid
modelled on the Lie algebroid $\tau_V:V\to Q$. Suppose that
$(x^i)$ are local coordinates on an open subset $U$ of $Q$ and
that $\{e_0,e_{\alpha}\}$ is a local basis of
$\Gamma(\tau_{\widetilde{{\mathcal A}}})$ in $U$ which is adapted
to the $1$-cocycle $1_{\mathcal A}$, i.e., such that $1_{\mathcal
A}(e_0)=1$ and $1_{\mathcal A}(e_{\alpha})=0,$ for all $\alpha.$
Note that if $\{e^0,e^{\alpha}\}$ is the dual basis of
$\{e_0,e_{\alpha}\}$ then $e^0=1_{\mathcal A}$. Moreover, since
$1_\A$ is a 1-cocycle, we have that
$$\begin{array}{rclrclccrclrcl} \lcf
e_0,e_{\alpha}\rcf_{\widetilde{{\mathcal
A}}}=C_{0\alpha}^{\gamma}e_{\gamma},\;\;\lcf
e_{\alpha},e_{\beta}\rcf_{\widetilde{{\mathcal
A}}}=C_{\alpha\beta}^{\gamma}e_{\gamma},\;\;
\rho_{\widetilde{{\mathcal
A}}}(e_0)=\rho_0^i\displaystyle\frac{\partial}{\partial x^i},\;\;
\rho_{\widetilde{{\mathcal
A}}}(e_{\alpha})=\rho_{\alpha}^i\displaystyle\frac{\partial}{\partial
x^i}.
\end{array}
$$
Denote by $(x^i,y^0,y^{\alpha})$ the corresponding local
coordinates on $\widetilde{{\mathcal A}}$. Then, the local
equation defining the affine subbundle ${\mathcal A}$
(respectively, the vector subbundle $V$) of $\widetilde{{\mathcal
A}}$ is $y^0=1$ (respectively, $y^0=0$). Thus, $(x^i,y^{\alpha})$
may be considered as local coordinates on ${\mathcal A}$ and $V$.

The standard example of a Lie affgebroid is the 1-jet bundle
$\tau_{1,0}:J^1\tau\to Q$ of local sections of a fibration
$\tau:Q\to\R$. It is well known that $\tau_{1,0}$ is an affine
bundle modelled on the vector bundle
$\pi=(\pi_Q)_{|V\tau}:V\tau\to Q$, where $V\tau$ is the vertical
bundle of $\tau$. Moreover, if $t$ is the usual coordinate on $\R$
and $\eta$ is the closed $1$-form on $Q$ given by
$\eta=\tau^*(dt)$ then we have the identification
$J^1\tau\cong\{v\in TQ\,|\,\eta(v)=1\}$ (see, for instance,
\cite{Sa}). Note that $V\tau=\{v\in TQ\,|\,\eta(v)=0\}.$ Thus, the
bidual bundle $\widetilde{J^1\tau}$ to $\tau_{1,0}:J^1\tau\to Q$
may be identified with the tangent bundle $TQ$ to $Q$ and, under
this identification, the Lie algebroid structure on $\pi_Q:TQ\to
Q$ is the standard Lie algebroid structure and the $1$-cocycle
$1_{J^1\tau}$ on $\pi_Q:TQ\to Q$ is just $\eta$.

\subsection{The Hamiltonian formalism}\label{sec2.5}


Suppose that $(\tau_{\mathcal A}:{\mathcal A}\rightarrow Q,$
$\tau_V:V\rightarrow Q)$ is a Lie affgebroid. Then, we consider
the prolongation ${\mathcal T}^{\widetilde{{\mathcal A}}}V^*$ of
the bidual Lie algebroid $(\widetilde{{\mathcal
A}},\lcf\cdot,\cdot\rcf_{\widetilde{{\mathcal A}}},$ $
\rho_{\widetilde{{\mathcal A}}})$ over the fibration
$\tau_{V}^*:V^*\rightarrow Q$ and denote by $(\lcf \cdot , \cdot
\rcf_{\widetilde{{\mathcal A}}}^{\tau_{V}^*},
\rho_{\widetilde{{\mathcal A}}}^{\tau_{V}^*})$ the Lie algebroid
structure on ${\mathcal T}^{\widetilde{{\mathcal A}}}V^*$ (for the
definition of the Lie algebroid structure on the prolongation of a
Lie algebroid over a fibration, we remit to \cite{HM,LMM}).

Let $\mu:{\mathcal A}^+\rightarrow V^*$ be the canonical
projection given by $\mu(\varphi)=\varphi^l$, for $\varphi\in
{\mathcal A}^+_x$, with $x\in Q$, where $\varphi^l\in V^*_x$ is
the linear map associated with the affine map $\varphi$ and
$h:V^*\rightarrow {\mathcal A}^+$ be a {\it Hamiltonian section}
of $\mu$, that is, $\mu \comp h = Id$.

Now, we consider the prolongation ${\mathcal
T}^{\widetilde{{\mathcal A}}}{\mathcal A}^+$ of the Lie algebroid
$\widetilde{{\mathcal A}}$ over $\tau_{{\mathcal A}^+}:{\mathcal
A}^+\to Q$ with vector bundle projection $\tau^{\tau_{{\mathcal
A}^+}}_{\widetilde{{\mathcal A}}}:{\mathcal
T}^{\widetilde{{\mathcal A}}}{\mathcal A}^+\rightarrow {\mathcal
A}^+$. Then, we may introduce the map ${\mathcal T}h:{\mathcal
T}^{\widetilde{{\mathcal A}}}V^*\rightarrow{\mathcal
T}^{\widetilde{{\mathcal A}}}{\mathcal A}^+$ defined by ${\mathcal
T}h(\tilde{\a},X_{\alpha})=(\tilde{\a},(T_{\alpha}h)(X_{\alpha})),$
for $(\tilde{\a},X_{\alpha})\in {\mathcal
T}_\alpha^{\widetilde{{\mathcal A}}}V^*$, with $\alpha\in V^*.$ It
is easy to prove that the pair $({\mathcal  T}h,h)$ is a Lie
algebroid morphism between the Lie algebroids
$\tau_{\widetilde{{\mathcal A}}}^{\tau_V^*}:{\mathcal
T}^{\widetilde{{\mathcal A}}}V^*\rightarrow V^*$ and
$\tau_{\widetilde{{\mathcal A}}}^{\tau_{{\mathcal A}^+}}:{\mathcal
T}^{\widetilde{{\mathcal A}}}{\mathcal A}^+\rightarrow {\mathcal
A}^+$. We denote by $\lambda_h$ and $\Omega_h$ the sections of the
vector bundles $({\mathcal T}^{\widetilde{{\mathcal
A}}}V^*)^*\rightarrow V^*$ and $\Lambda^2({\mathcal
T}^{\widetilde{{\mathcal A}}}V^*)^*\rightarrow V^*$ defined by
\begin{equation}\label{Omegah}
\lambda_h=({\mathcal T}h,h)^*(\lambda_{\widetilde{{\mathcal
A}}}),\;\;\Omega_h=({\mathcal
T}h,h)^*(\Omega_{\widetilde{{\mathcal A}}}),
\end{equation}
$\lambda_{\widetilde{{\mathcal A}}}$ and
$\Omega_{\widetilde{{\mathcal A}}}$ being the Liouville section
and the canonical symplectic section, respectively, associated
with the Lie algebroid $\widetilde{{\mathcal A}}$ (see
\cite{LMM}). Note that $\Omega_h=-d^{{\mathcal
T}^{\widetilde{{\mathcal A}}}V^*}\lambda_h.$

On the other hand, let $\eta:{\mathcal T}^{\widetilde{{\mathcal
A}}}V^*\rightarrow\R$ be the section of $({\mathcal
T}^{\widetilde{{\mathcal A}}}V^*)^*\rightarrow V^*$ given by
\begin{equation}\label{eta}
\eta(\tilde{\a},X_{\nu})=1_{\mathcal A}(\tilde{\a}),\,\mbox{ for
}\,(\tilde{\a},X_{\nu})\in {\mathcal T}_\nu^{\widetilde{{\mathcal
A}}}V^*,\,\mbox{ with }\,\nu\in V^*.
\end{equation}
We remark that if $\pr_1:{\mathcal T}^{\widetilde{{\mathcal
A}}}V^*\to \widetilde{{\mathcal A}}$ is the canonical projection
on the first factor then $(\pr_1,\tau_V^*)$ is a morphism between
the Lie algebroids $\tau_{\widetilde{{\mathcal
A}}}^{\tau_V^*}:{\mathcal T}^{\widetilde{{\mathcal
A}}}V^*\rightarrow V^*$ and $\tau_{\widetilde{{\mathcal
A}}}:\widetilde{{\mathcal A}}\to Q$ and
$(\pr_1,\tau_V^*)^*(1_{\mathcal A})=\eta$. Thus, since
$1_{\mathcal A}$ is a $1$-cocycle of $\tau_{\widetilde{{\mathcal
A}}}:\widetilde{{\mathcal A}}\rightarrow Q$, we deduce that $\eta$
is a $1$-cocycle of the Lie algebroid $\tau_{\widetilde{{\mathcal
A}}}^{\tau_V^*}:{\mathcal T}^{\widetilde{{\mathcal
A}}}V^*\rightarrow V^*.$

Let $(x^i)$ be local coordinates on an open subset $U$ of $Q$ and
$\{e_0,e_{\alpha}\}$ be a local basis of
$\Gamma(\tau_{\widetilde{{\mathcal A}}})$ on $U$ adapted to
$1_{\mathcal A}$. Denote by $(x^i,y^0,y^{\alpha})$ the induced
local coordinates on $\widetilde{{\mathcal A}}$ and by
$(x^i,y_0,y_{\alpha})$ the dual coordinates on ${\mathcal A}^+$.
Then, $(x^i,y_{\alpha})$ are local coordinates on $V^*$ and
$\{\Y_0,\Y_{\alpha},\U_{\alpha}\}$ is a local basis of
$\Gamma(\tau^{\tau_{V}^*}_{\widetilde{{\mathcal A}}})$, where
$$\Y_0(\psi)=\Big(e_0(x),\rho_0^i\displaystyle\frac{\partial}{\partial
x^i}_{|\psi}\Big),\Y_{\alpha}(\psi)=\Big(e_{\alpha}(x),\rho_{\alpha}^i\displaystyle\frac{\partial}{\partial
x^i}_{|\psi}\Big),\U_{\alpha}(\psi)=\Big(0,\displaystyle\frac{\partial}{\partial
y_{\alpha}}_{|\psi}\Big),$$ for $\psi\in V^*_x$. Suppose that
$h(x^i,y_{\alpha})=(x^i,-H(x^j,y_{\beta}),y_{\alpha})$ and that
$\{\Y^0,\Y^{\alpha},\U^{\alpha}\}$ is the dual basis of
$\{\Y_0,\Y_{\alpha},$ $\U_{\alpha}\}$. Then $\eta=\Y^0$ and, from
(\ref{Omegah}) and the definition of the map ${\mathcal T}h$, it
follows that
$$\Omega_h=\Y^{\alpha}\wedge\U^{\alpha}+\frac{1}{2}C_{\alpha\beta}^{\gamma}y_{\gamma}\Y^{\alpha}\wedge\Y^{\beta}+
\Big(\rho_{\alpha}^i\frac{\partial H}{\partial
x^i}-C_{0\alpha}^{\gamma}y_{\gamma}\Big)\Y^{\alpha}\wedge\Y^0+\frac{\partial
H }{\partial y_{\alpha}}\U^{\alpha}\wedge\Y^0.$$

Thus, it is easy to prove that the pair $(\Omega_h,\eta)$ is a
cosymplectic structure on the Lie algebroid
$\tau_{\widetilde{{\mathcal A}}}^{\tau_V^*}:{\mathcal
T}^{\widetilde{{\mathcal A}}}V^*\rightarrow V^*$ (this means that
$d^{{\mathcal T}^{\widetilde{\mathcal A}}V^*}\Omega_{h} = 0$,
$d^{{\mathcal T}^{\widetilde{\mathcal A}}V^*}\eta = 0$ and
$(\eta\wedge\Omega_{h}\wedge \dots^n\dots \wedge \Omega_{h})(\psi)
\neq 0$, $\forall \psi \in V^*$, where $n$ is the rank of
${\mathcal A}$). If $R_h\in\Gamma(\tau_{\widetilde{{\mathcal
A}}}^{\tau_V^*})$ is the Reeb section of $(\Omega_h,\eta)$ (that
is, $i_{{\mathcal R}_{h}}\Omega_{h} = 0$ and $i_{{\mathcal
R}_{h}}\eta = 1$), then its integral curves (i.e., the integral
curves of $\rho_{\widetilde{{\mathcal A}}}^{\tau_{V}^*}(R_h)$) are
just {\it the solutions of the Hamilton equations for $h$},
$$\frac{dx^i}{dt}=\rho_0^i+\rho_{\alpha}^i\frac{\partial
H}{\partial
y_{\alpha}},\;\;\;\frac{dy_{\alpha}}{dt}=-\rho_{\alpha}^i\frac{\partial
H }{\partial
x^i}+y_{\gamma}\Big(C_{0\alpha}^{\gamma}+C_{\beta\alpha}^{\gamma}\frac{\partial
H }{\partial y_{\beta}}\Big),$$
for $i\in\{1,\dots,m\}$ and $\alpha\in\{1,\dots,n\}$.


Next, we will present an alternative approach in order to obtain
the Hamilton equations. For this purpose, we will use the notion
of an aff-Poisson structure on an AV-bundle which was introduced
in \cite{GGrU} (see also \cite{GGU2}).

Let $\tau_Z:Z\to Q$ be an affine bundle of rank $1$ modelled on
the trivial vector bundle $\tau_{Q\times\R}:Q\times\R\to Q$, that
is, $\tau_Z:Z\to Q$ is an {\em AV-bundle } in the terminology of
\cite{GGU2}. Then, we have an action of $\R$ on the fibres of $Z$.
This action induces a vector field $X_Z$ on $Z$ which is vertical
with respect to the projection $\tau_Z:Z\to Q$.

On the other hand, there exists a one-to-one correspondence
between the space of sections of $\tau_Z:Z\to Q$,
$\Gamma(\tau_Z)$, and the set $\{F_h\in
C^{\infty}(Z)\,|\,X_Z(F_h)=-1\}.$ In fact, if $h\in\Gamma(\tau_Z)$
and $(x^i,s)$ are local fibred coordinates on $Z$ such that
$X_Z=\displaystyle\frac{\partial}{\partial s}$ and $h$ is locally
defined by $h(x^i)=(x^i,-H(x^i))$, then the function $F_h$ on $Z$
is locally given by $F_h(x^i,s)=-H(x^i)-s,$ (for more details, see
\cite{GGU2}).

Now, an {\em aff-Poisson structure } on the AV-bundle $\tau_Z:Z\to
Q$ is a bi-affine map,
$\{\cdot,\cdot\}:\Gamma(\tau_Z)\times\Gamma(\tau_Z)\to
C^{\infty}(Q)$, which satisfies the following properties:
\begin{enumerate}
\item[i)] Skew-symmetric: $\{h_1,h_2\}=-\{h_2,h_1\}$.
\item[ii)] Jacobi identity: $\{h_1,\{h_2,h_3\}\}_V+\{h_2,\{h_3,h_1\}\}_V+\{h_3,\{h_1,h_2\}\}_V=0,$
where $\{\cdot,\cdot\}_V$ is the affine-linear part of the
bi-affine bracket.
\item[iii)] If $h\in\Gamma(\tau_Z)$ then the map $\{h,\cdot\}:\Gamma(\tau_Z)\to
C^{\infty}(Q)$ defined by $\{h, \cdot \}(h') = \{h, h'\}$, for $h'
\in \Gamma(\tau_{Z})$, is an affine derivation.
\end{enumerate}
Condition iii) implies that, for each $h\in\Gamma(\tau_Z)$ the
linear part $\{h,\cdot\}_V:C^{\infty}(Q)\to C^{\infty}(Q)$ of the
affine map $\{h,\cdot\}:\Gamma(\tau_Z)\to C^{\infty}(Q)$ defines a
vector field on $Q$, which is called {\it the Hamiltonian vector
field of $h$} (see \cite{GGU2}).

In \cite{GGU2}, the authors proved that there is a one-to-one
correspondence between aff-Poisson bra\-ckets $\{\cdot,\cdot\}$ on
$\tau_Z:Z\to Q$ and Poisson brackets $\{\cdot,\cdot\}_{\Pi}$ on
$Z$ which are $X_Z$-invariant, i.e., which are associated with
Poisson 2-vectors $\Pi$ on $Z$ such that ${\mathcal
L}_{X_Z}\Pi=0$. This correspondence is determined by
$$\{h_1,h_2\}\comp\tau_Z=\{F_{h_1},F_{h_2}\}_\Pi,\makebox[1cm]{for}h_1,h_2\in\Gamma(\tau_Z).$$

Using this correspondence one may prove the following result.

\begin{theorem}\label{teor2.2}\cite{IMPS} Let
$\tau_{\mathcal A}:{\mathcal A}\to Q$ be a Lie affgebroid modelled
on the vector bundle $\tau_V:V\to Q$. Denote by $\tau_{{\mathcal
A}^+}:{\mathcal A}^+\to Q$ (resp., $\tau_V^*:V^*\to Q$) the dual
vector bundle to ${\mathcal A}$ (resp., to $V$) and by
$\mu:{\mathcal A}^+\to V^*$ the canonical projection. Then:
\begin{enumerate}
\item[i)] $\mu:{\mathcal A}^+\to V^*$ is an AV-bundle which admits
an aff-Poisson structure. \item[ii)] If $h:V^*\to {\mathcal A}^+$
is a Hamiltonian section then the Hamiltonian vector field of $h$
with respect to the aff-Poisson structure is a vector field on
$V^*$ whose integral curves are just the solutions of the Hamilton
equations for $h$.
\end{enumerate}
\end{theorem}

\section{Vakonomic mechanics on Lie affgebroids}

\subsection{Vakonomic equations and vakonomic bracket}\label{sec4.1}

Let $\tau_{\A}:\A\rightarrow Q$ be a Lie affgebroid of rank $n$
over a manifold $Q$ of dimension $m$. We consider an embedded
submanifold $\M\subseteq\A$, called {\it the constraint
submanifold}, of dimension $n+m-\bar{m}$ such that
$\tau_\M=\tau_{\A|\M}:\M\to Q$ is a surjective submersion.

Now, suppose that $e$ is a point of $\M$, with $\tau_\M(e)=x$,
that $(x^i)$ are local coordinates on an open subset $U$ of $Q$,
$x\in U$, and that $\{e_0,e_\alpha\}$ is a local basis of
$\Gamma(\tau_\bidual)$ on $U$ adapted to the $1$-cocycle $1_\A$.
Denote by $(x^i,y^0,y^\alpha)$ (respectively, $(x^i,y^\alpha)$)
the corresponding local coordinates for $\bidual$ (respectively,
$\A$) on the open subset $\tau^{-1}_\bidual(U)$ (respectively,
$\tau^{-1}_\A(U)$). Assume that
$$\M\cap\tau^{-1}_\A(U)\equiv\{(x^i,y^\alpha)\in\tau^{-1}_\A(U) \,|\,
\Phi^A(x^i,y^\alpha)=0,\;A=1,\dots,\bar{m}\}.$$ The rank of the
$(\bar{m}\times(n+m))$-matrix
$\Big(\displaystyle\frac{\partial\Phi^A}{\partial
x^i},\frac{\partial\Phi^A}{\partial y^\alpha}\Big)$ is maximun,
that is, $\bar{m}$. Then, using that $\tau_\M:\M\to Q$ is a
submersion, we can suppose that the
$(\bar{m}\times\bar{m})$-matrix
$$\Big(\displaystyle{\frac{\partial\Phi^A}{\partial
y^B}}_{|e}\Big)_{A=1,\dots,\bar{m};B=1,\dots,\bar{m}}$$ is
regular. Then, we will use the following notation
$(y^\alpha)=(y^A,y^a),$ for $1\leq \alpha\leq n$, $1\leq
A\leq\bar{m}$ and $\bar{m}+1\leq a\leq n$.

Now, using the implicit function theorem, we obtain that there
exist an open subset $\widetilde{V}$ of $\tau^{-1}_\A(U)$, an open
subset $W\subseteq\R^{m+n-\bar{m}}$ and smooth real functions
$\Psi^A:W\to\R,\;\;A=1,\dots,\bar{m},$ such that
$$\M\cap\widetilde{V} \equiv \{(x^i,y^\alpha)\in\widetilde{V} \,|\,
y^A=\Psi^A(x^i,y^a),\;A=1,\dots,\bar{m}\}.$$ Consequently,
$(x^i,y^a)$ are local coordinates on $\M$.

Next, consider the Whitney sum of $\A^+$ and $\A$, that is,
$\A^+\oplus_Q\A$ and the canonical projections
$\pr_1:\A^+\oplus_Q\A\to\A^+$ and $\pr_2:\A^+\oplus_Q\A\to\A$. Let
$W_0$ be the submanifold of $\A^+\oplus_Q\A$ given by
$W_0=\pr_2^{-1}(\M)=\A^+\oplus_Q\M$ and the restrictions
$\pi_1={\pr_1}_{|W_0}$ and $\pi_2={\pr_2}_{|W_0}$. Also denote by
$\nu: W_0\to Q$ the canonical projection.

%

\vspace{.2cm} Now, we take the prolongation
$\tau_\bidual^{\tau_{\A^+}}:{\mathcal T}^\bidual \A^+\to \A^+$
(respectively, $\tau_\bidual^\nu:{\mathcal T}^\bidual W_0\to W_0$)
of the Lie algebroid $\bidual$ over $\tau_\dual:\dual\to Q$
(res\-pectively, $\nu :W_0\to Q$). Moreover, we can prolong $\pi
_1:W_0\to\dual$ to a morphism of Lie algebroids ${\mathcal T}\pi
_1 :{\mathcal T}^\bidual W_0\to {\mathcal T}^\bidual\dual$ defined
by $\T\pi_1=(Id,T\pi_1)$.

If $(x^i,y_0,y_\alpha)$ are the local coordinates on $\dual$
induced by the dual basis $\{e^0,e^\alpha\}$ of the local basis
$\{e_0,e_\alpha\}$ of $\Gamma(\tau_\bidual)$, then
$(x^i,y_0,y_\alpha,y^a)$ are local coordinates for $W_0$ and we
may consider the local basis $\{\Y_0,{\mathcal Y}_\alpha
,\P^0,{\mathcal P}^\alpha ,\mathcal{V}_a\}$ of
$\Gamma(\tau_\bidual^\nu)$ defined by
$$\begin{array}{c} {\mathcal Y}_0(\varphi, \a)=\Big(e_0(x),
\rho_0^i\displaystyle\frac{\partial }{\partial x^i}_{|\varphi},
0\Big),\;\;\; {\mathcal Y}_\alpha(\varphi, \a)=\Big(e_\alpha(x),
\rho_\alpha^i\displaystyle\frac{\partial }{\partial
x^i}_{|\varphi}, 0\Big),\\[12pt]
{\mathcal P}^0(\varphi, \a)=\Big(0,\displaystyle\frac{\partial
}{\partial y_{0}}_{|\varphi}, 0\Big),\; {\mathcal
P}^\alpha(\varphi, \a)=\Big(0,\displaystyle\frac{\partial
}{\partial y_{\alpha}}_{|\varphi}, 0\Big),\; {\mathcal
V}_a(\varphi, \a)=\Big(0, 0,\displaystyle\frac{\partial }{\partial
y^{a}}_{|\a}\Big),
\end{array}$$
for $(\varphi, \a)\in W_0$ and $\nu(\varphi,\a)=x$, where
$\rho_0^i$ and $\rho_\alpha^i$ are the components of the anchor
map $\rho_\bidual$ with respect to the local basis
$\{e_0,e_\alpha\}$.

Now, one may consider on the Lie algebroid
$\tau_\bidual^\nu:\T^\bidual W_0\to W_0$ the presymplectic
2-section $\Omega =({\mathcal T}\pi _1, \pi_1)^*\Omega _\bidual,$
where $\Omega _\bidual$ is the canonical symplectic section on
${\mathcal T}^\bidual\dual$. The local expression of $\Omega$ is
\begin{equation}\label{locOmega}
\Omega=\Y^0\wedge\P_0+\Y^\alpha\wedge {\mathcal P}_\alpha
+C_{0\alpha}^\gamma y_\gamma\Y^0\wedge\Y^\alpha+ \frac{1}{2}
{C}_{\alpha\beta}^\gamma y_\gamma \Y^\alpha\wedge \Y^\beta,
\end{equation}
$\{\Y^0,{\mathcal Y}^\alpha ,\P_0,{\mathcal P}_\alpha
,\mathcal{V}^a\}$ being the dual basis of the local basis
$\{\Y_0,{\mathcal Y}_\alpha ,\P^0,{\mathcal P}^\alpha
,\mathcal{V}_a\}$.

On the other hand, if $\pr_1:\T^\bidual W_0\to\bidual$ is the
canonical projection on the first factor, then we can introduce
the section $\eta\in\Gamma((\tau_\bidual^\nu)^*)$ defined by
$\eta=(\pr_1,\nu)^*1_\A.$ Since $1_\A$ is a 1-cocycle of
$\bidual\to Q$, we deduce that $\eta$ is a 1-cocycle of
$\T^\bidual W_0\to W_0$. Moreover, it is easy to prove that
\begin{equation}\label{etaloc} \eta=\Y^0.
\end{equation}

Now, let $L:\A\to\R$ be a Lagrangian function on $\A$ and denote
by $\tilde{L}$ the restriction of $L$ to the constraint
submanifold $\M$.

{\it The Pontryagin Hamiltonian} $H_{W_0}$ is the real function in
$W_0=\dual\oplus _Q \M$ given by $H_{W_0} (\varphi , \a)= \varphi
(\a) -\tilde{L}(\a),$ or, in local coordinates,
\begin{equation}\label{HW0locaf}
H_{W_0}(x^i,y_0, y_\alpha, y^a)=y_0+y_ay^a+y_{A}\Psi^{A}(x^i,
y^a)-\tilde{L}(x^i, y^a)\, .
\end{equation}
Thus, one can consider the presymplectic 2-section $\Omega_{W_0}$
on $\T^\bidual W_0$ defined by
$$\Omega_{W_0}=\Omega+d^{\T^\bidual W_0}H_{W_0}\wedge\eta.$$
In local coordinates, using (\ref{locOmega}), (\ref{etaloc}) and
(\ref{HW0locaf}), we deduce that
\begin{equation}\label{OmegaHW0loc}
\begin{array}{l}
\Omega_{W_0}=\Y^\alpha\wedge {\mathcal
P}_\alpha+\Big[\Big(y_A\displaystyle\frac{\partial \Psi
^A}{\partial x^i} -\frac{\partial \tilde{L}}{\partial
x^i}\Big)\rho_\alpha^i+C_{\alpha 0}^\gamma
y_\gamma\Big]\Y^\alpha\wedge\Y^0+y^a
\P_a\wedge\Y^0\;\;\;\;\;\\[12pt]
\hfill+\Psi ^A \P_A\wedge\Y^0+ \displaystyle\frac{1}{2}
{C}_{\alpha\beta}^\gamma y_\gamma \Y^\alpha\wedge \Y^\beta +
\Big(y_a+y_A \displaystyle\frac{\partial \Psi ^A}{\partial
y^a}-\frac{\partial \tilde{L}}{\partial y^a}\Big)\V^a\wedge\Y^0.
\end{array}\end{equation}

\begin{definition}{\rm The vakonomic problem} $(L,\M)$ on the Lie affgebroid $\A$
consists of to find the solutions for the equations
\begin{equation}\label{problvakafg}
i_{X}\Omega _{W_0} =0\;\mbox{ and }\;i_X\eta=1,\;\mbox{ with
}\;X\in\Gamma(\tau_\bidual^\nu).
\end{equation}

\end{definition}

First, we will obtain the local expression of the vakonomic
problem. In general, a section $X$ satisfying the equations
(\ref{problvakafg}) cannot be found in all points of $W_0$. Thus,
we consider the points where (\ref{problvakafg}) have sense. We
define
$$ W_1=\{ w\in W_0 \, | \, \exists\,
 Z\in\T^\bidual_w W_0: i_{Z}\Omega_{W_0}(w)=0\,\mbox{ and }\,i_Z\eta(w)=1\}.$$

In local coordinates, we deduce that $W_1$ is characterized by the
equations
$$\varphi _a = y_a+y_A \frac{\partial \Psi ^A}{\partial
y^a}-\frac{\partial \tilde{L}}{\partial y^a}=0, \quad
\bar{m}+1\leq a \leq n .$$

Moreover, a direct computation, using (\ref{etaloc}) and
(\ref{OmegaHW0loc}), proves that the local expression of any
section $X$ satisfying the equations (\ref{problvakafg}) is of the
form
\[\begin{array}{rcl}
X_{(\Upsilon_0,\Upsilon^a)}&=&\Y_0+\Psi ^A {\mathcal Y}_A+
y^a{\mathcal Y}_a + \Upsilon_0\P^0 + \Big [ \rho ^i_\alpha \Big (
\displaystyle\frac{\partial \tilde{L}}{\partial x^i}-y_A
\frac{\partial \Psi ^A}{\partial x^i}
\Big )\\[12pt]
&&-y_\gamma(C_{\alpha 0}^\gamma+\Psi ^A {C}^\gamma_{\alpha A}
+y^a{C}^\gamma_{\alpha a}) \Big ]{\mathcal P}^\alpha + \Upsilon
^a{\mathcal V}_a, \end{array}
\]
with $\Upsilon_0$ and $\Upsilon^a$ arbitrary functions.
Consequently, the vakonomic equations are

\begin{equation}\label{eqvakafg}\left \{
\begin{array}{l}
\displaystyle \dot{x}^i=\rho_0^i+\Psi ^A \rho ^i_A+y^a \rho ^i_a ,\\[4pt]
\displaystyle \dot{y}_{A}=\Big ( \frac{\partial
\tilde{L}}{\partial x^i}-y_B \frac{\partial \Psi ^B}{\partial x^i}
\Big ) \rho ^i_A -y_\gamma(C^{\gamma}_{A0}+\Psi ^B C^\gamma_{AB}+y^aC^\gamma_{Aa}) ,\\[9pt]
\displaystyle \frac{d}{dt}\left( \frac{\partial
\tilde{L}}{\partial y^a}-y_A \frac{\partial \Psi ^A}{\partial y^a}
\right)= \Big ( \frac{\partial \tilde{L}}{\partial x^i}-y_A
\frac{\partial \Psi ^A}{\partial x^i} \Big ) \rho ^i_{a}\\[12pt]
\hspace{4.5cm}-y_\gamma( C^{\gamma}_{a0}+\Psi ^B
C^\gamma_{aB}+y^bC^\gamma_{ab}),
\end{array}
\right .\end{equation} for all $1\leq i\leq m$, $1\leq
A\leq\bar{m}$ and $\bar{m}+1\leq a\leq n$.

\begin{remark}{\rm The motion equations for the vakonomic
mechanics may be also expressed as follows
\begin{equation}\label{eqvakafg2}\left \{
\begin{array}{l}\dot{x}^i=\rho_0^i+y^\alpha \rho ^i_\alpha ,\\[4pt]
\displaystyle \frac{d}{dt}\left( \frac{\partial {L}}{\partial
y^\alpha} \right)-\rho_\alpha^i\frac{\partial L}{\partial
x^i}=-\lambda_A \Big [\frac{d}{dt}\Big( \frac{\partial
\phi^A}{\partial y^\alpha}\Big)-\rho_\alpha^i \frac{\partial \phi
^A}{\partial x^i} \Big
]-\dot{\lambda_A}\frac{\partial\phi^A}{\partial y^\alpha}\\[10pt]
\hspace{4.07cm}-y_\gamma(C_{\alpha0}^\gamma+y^\beta C_{\alpha\beta}^\gamma)\\[4pt]
\phi^A=0,
\end{array}
\right .\end{equation} where $\phi^A=y^A-\Psi^A$ and
$\lambda_A=y_A-\frac{\partial L}{\partial y^A}$. Note that in
contrast to equations (\ref{eqvakafg}), equations
(\ref{eqvakafg2}) are expressed in terms of the global Lagrangian
$L:\A\to\R$. Thus, the equations (\ref{eqvakafg}) stress how the
information given by the Lagrangian $L$ outside $\M$ is irrelevant
to obtain the vakonomic equations. This is in contrast with what
happens in nonholonomic mechanics (see \cite{IMMS}). }
\end{remark}


Then, we know that there exist sections $X$ of ${\mathcal
T}^{\bidual}W_{0|W_1}\to W_1$ satisfying (\ref{problvakafg}).
However, $X$ doesn't belong, in general, to ${\mathcal
T}^{\bidual}W_1\subseteq \bidual\times TW_1$. In fact, one may
prove that the restriction to $W_1$ of
$X_{(\Upsilon_0,\Upsilon^a)}$ is a section of $\T^\bidual W_1\to
W_1$ if and only if
$$[\Upsilon^a(d^{\T^\bidual W_0}\varphi_b)(\V_a)=(d^{\T^\bidual
W_0}\varphi_b)(\Upsilon^a\V_a-X_{(\Upsilon_0,\Upsilon^a)})]_{|W_1},\;\forall\,
b.$$ Then we have a system of $(n-\bar{m})$ equations with
$(n-\bar{m})$ unknowns (the functions $\Upsilon^a$). Thus, if we
denote by ${\mathcal R}_{ab}$ and $\mu_b$ the functions
$$\begin{array}{rcl}
{\mathcal R}_{ab}&=&[(d^{\T^\bidual
W_0}\varphi_b)(\V_a)]_{|W_1}=\left(\displaystyle\frac{\partial^2\tilde{L}}{\partial
y^a\partial y^b}-y_A\frac{\partial^2\Psi^A}{\partial y^a\partial
y^b}\right)_{|W_1},\\
\mu_b&=&[(d^{\T^\bidual
W_0}\varphi_b)(\Upsilon^a\V_a-X_{(\Upsilon_0,\Upsilon^a)})]_{|W_1},
\end{array}$$
it is clear that the above system has a solution $\Upsilon^a$ if
the matrices ${\mathcal R}=({\mathcal R}_{ab})$ and ${\mathcal
R}_\mu=({\mathcal R}_{ab};\mu_b)$ have the same rank. Note that
even if the above system has a unique solution (i.e., if the
matrix ${\mathcal R}=({\mathcal R}_{ab})$ is regular), the
solution $\left(X_{(\Upsilon_0,\Upsilon^a)}\right)_{|W_1}$ is not,
in general, unique (since the function $(\Upsilon_0)_{|W_1}$ is
still arbitrary).

To solve the above problem, we consider a suitable submanifold
$W_1'$ of $W_1$ whose intrinsic definition is
$$W_1'=\{w\in W_1\,|\, H_{W_1}(w)=0\},$$
where $H_{W_1}:W_1\to\R$ is the restriction to $W_1$ of the
Pontryagin Hamiltonian $H_{W_0}$. In local coordinates, the
submanifold $W_1'$ is given by the equation
\begin{equation}\label{eqdefW1'}
y_0+y_A\Psi^A(x^i,y^b)+y_a y^a-\tilde{L}(x^i,y^b)=0.
\end{equation}

Let $\Omega_{W_1'}$ (respectively, $\eta_{W_1'}$) be the
restriction of $\Omega_{W_0}$ (respectively, $\eta$) to ${\mathcal
T}^{\bidual}W_1'$. Note that the restriction $\nu_1':W_1'\to Q$ of
$\nu:W_0\to Q$ to $W_1'$ is a fibration and, therefore, we can
consider the prolongation ${\mathcal T}^{\bidual}W_1'$ of the Lie
algebroid $\bidual$ over $\nu_1'$. Moreover, we have the following
result.

\begin{proposition}\label{prop6.15} $(\Omega_{W_1'},\eta_{W_1'})$
is a cosymplectic structure on ${\mathcal T}^{\bidual}W_1'$ if and
only if for any system of coordinates $(x^i,y_0, y_\alpha, y^a)$
on $W_0$ we have that
\[
\det({\mathcal R}_{ab})=\det\left( \frac{\partial^2
\tilde{L}}{\partial y^a
\partial y^b}-y_A\frac{\partial^2 {\Psi^{A}}}{\partial y^a
\partial y^b}\right)\not=0,\;\mbox{ for all point in }W_1'.
\]
\end{proposition}
\begin{proof} It is clear that $d^{{\mathcal
T}^{\bidual}W_1'}\Omega_{W_1'}=0$ and $d^{{\mathcal
T}^{\bidual}W_1'}\eta_{W_1'}=0$.

Now, suppose that the matrix $({\mathcal R}_{ab})$ is regular.
Since the rank of the Lie algebroid $\T^\bidual W_1'\to W_1'$ is
$(2n+1)$, we have to prove that
$$\ker\Omega_{W_1'}(w_1')\cap\ker\eta_{W_1'}(w_1')=\{0\},\;\forall\,
w_1'\in W_1'.$$ Now, let
$Z\in\ker\Omega_{W_1'}(w_1')\cap\ker\eta_{W_1'}(w_1').$ {}From
(\ref{OmegaHW0loc}), it follows that
$$(i_Z\Omega_{W_0}(w_1'))(\P^0(w_1'))=(i_Z\Omega_{W_0}(w_1'))(\V_a(w_1'))=0,\;\mbox{
for all }\;a.$$ On the other hand,
$$(d^{\T^\bidual W_0}H_{W_0})(w_1')(\P^0(w_1'))=1,\;\;(d^{\T^\bidual
W_0}\varphi_b)(w_1')(\V_a(w_1'))={\mathcal R}_{ab}(w_1'),\,\mbox{
for all }\,b.$$ Thus,
$\P^0(w_1')\not\in\T^\bidual_{w_1'}W_1',\;\;\V_a(w_1')\not\in\T^\bidual_{w_1'}W_1'$
and $Z\in\ker\Omega_{W_0}(w_1')\cap\ker\eta(w_1').$ This implies
that $Z=\lambda_0 \P^0(w_1')+\lambda^a \V_a(w_1').$ Therefore,
since $Z\in\T^\bidual_{w_1'}W_1'$, we have that
$$\begin{array}{c}
0=\lambda_0 (d^{\T^\bidual
W_0}\varphi_b)(w_1')(\P^0(w_1'))+\lambda^a (d^{\T^\bidual
W_0}\varphi_b)(w_1')(\V_a(w_1'))=\lambda^a {\mathcal
R}_{ab}(w_1'),
\end{array}$$
for all $b$, and, consequently, $\lambda^a=0$, for all $a$. Thus,
$Z=\lambda_0 \P^0(w_1')$ and
$$0=\lambda_0 (d^{\T^\bidual
W_0}H_{W_0})(w_1')(\P^0(w_1'))=\lambda_0,$$ that is, $Z=0$.

The converse is proved in a similar way.
\end{proof}



\begin{remark}\label{obs6.18}{\rm
We remark that the condition $\det\left( {\mathcal R}_{ab}
\right)\not=0$ implies that the matrix $\left( \displaystyle
\frac{\partial \varphi_{a}}{\partial y^{b}} \right)_{a, b =
\bar{m}+1, \dots , n}$ is regular. Thus, using the implicit
theorem function, we deduce that there exist open subsets
$\bar{W}_{0} \subseteq W_{0}$, $\tilde{W} \subseteq \R^{m+n}$ and
smooth real functions $\mu^{a}: \tilde{W} \to \R,$ $a = \bar{m}
+1, \dots, n,$ such that $W_{1} \cap \bar{W}_0$ is locally defined
by the equations
$$\begin{array}{rcl} y^{a} &=& \mu^{a}(x^i,
y_{\alpha}), \;\; a = \bar{m} + 1, \dots, n.
\end{array}$$
Therefore, we may consider $(x^i,y_0,y_\alpha)$ as local
coordinates on $W_1$ and, consequently, from (\ref{eqdefW1'}), we
obtain that $(x^i,y_\alpha)$ are local coordinates on $W_1'$.
Thus, a local basis of sections of ${\mathcal
T}^{\bidual}W_{1}'\to W_1'$ is given by $\{\Y_{01'},\Y_{\alpha
1'},\P_{1'}^\alpha\}$, where
$$\begin{array}{c}
{\mathcal Y}_{0 1'} = \Big({\mathcal
Y}_{0}+\rho_0^i\Big(\displaystyle\frac{\partial\tilde{L}}{\partial
x^i}-y_A\frac{\partial\Psi^A}{\partial x^i}\Big)\P^0 +
\rho^{i}_{0} \displaystyle \frac{\partial \mu^{a}}{\partial x^{i}}
{\mathcal V}_{a}
\Big)_{|W_{1}'},\\[12pt]
 {\mathcal Y}_{\alpha 1'} = \Big({\mathcal Y}_{\alpha} +\rho_\alpha^i\Big(\displaystyle\frac{\partial\tilde{L}}{\partial
x^i}-y_A\frac{\partial\Psi^A}{\partial x^i}\Big)\P^0+
\rho^{i}_{\alpha} \displaystyle \frac{\partial \mu^{a}}{\partial
x^{i}} {\mathcal V}_{a}
\Big)_{|W_{1}'},\\[12pt]
{\mathcal P}_{1'}^{A} =\Big({\mathcal P}^{A}-\Psi^A\P^0 +
\displaystyle \frac{\partial \mu^{a}}{\partial y_{A}} {\mathcal
V}_{a}\Big)_{|W_{1}'},\;\; {\mathcal P}_{1'}^{a} = \Big({\mathcal
P}^{a} -\mu^a\P^0+ \displaystyle \frac{\partial \mu^{b}}{\partial
y_{a}} {\mathcal V}_{b}\Big)_{|W_{1}'}.
\end{array}
$$
}
\end{remark}

Proceeding as in the proof of Proposition \ref{prop6.15}, we
deduce the following result.

\begin{theorem}\label{th6.19}
If $(\Omega_{W_1'},\eta_{W_1'})$ is a cosymplectic structure on
the Lie algebroid $\tau_\bidual^{\nu_1'}:{\mathcal
T}^{\bidual}W_{1}' \to W_{1}'$ then there exists a unique section
$\zeta_{1}\in\Gamma(\tau_\bidual^{\nu_1'})$ solution of the
vakonomic problem $(L, \M)$. In fact, $\zeta_{1}$ is the Reeb
section of $(\Omega_{W_1'},\eta_{W_1'})$, that is, $\zeta_1$ is
characterized by the conditions $i_{\zeta_{1}}\Omega_{W_1'}
=0\;\mbox{ and }\;i_{\zeta_1}\eta_{W_1'}=1.$
\end{theorem}



The above results suggest us to introduce the following
definition.

\begin{definition}
The vakonomic system $(L, \M)$ on the Lie affgebroid $\tau_\A: \A
\to Q$ is said to be {\rm regular} if the pair
$(\Omega_{W_1'},\eta_{W_1'})$ is a cosymplectic structure on the
Lie algebroid $\tau_\bidual^{\nu_1'}:{\mathcal T}^{\bidual}W_{1}'
\to W_{1}'$.
\end{definition}

In what follows, we will suppose that $(L,\M)$ is a regular
vakonomic system on the Lie affgebroid $\A$. Then, from Theorem
\ref{th6.19}, we have that the vakonomic problem has a unique
solution which is the Reeb section $\zeta_1$ of the cosymplectic
structure $(\Omega_{W_1'},\eta_{W_1'})$.

First, we will give the local expression of the solution section
$\zeta_1$. Suppose that $(x^{i}, y_{\alpha})$ are local
coordinates on $W_{1}'$ as in Remark \ref{obs6.18} and that
$\{\Y_{01'}, {\mathcal Y}_{\alpha 1'}, {\mathcal P}^{\alpha}_{1'}
\}$ is the corresponding local basis of
$\Gamma(\tau_\bidual^{\nu_1'})$. Then, if $\{\Y^{0}_{1'},
{\mathcal Y}^{\alpha}_{1'}, {\mathcal P}_{\alpha 1'}\}$ is the
dual basis of $\{\Y_{01'}, {\mathcal Y}_{\alpha 1'}, {\mathcal
P}^{\alpha}_{1'} \}$, we have that (see (\ref{OmegaHW0loc}))
$$\begin{array}{rcl} \Omega_{W_1'}& = &{\mathcal Y}^{\alpha}_{1'}
\wedge {\mathcal P}_{\alpha 1'} + \displaystyle \frac{1}{2}
{C}_{\alpha\beta}^{\gamma}y_{\gamma} {\mathcal Y}^{\alpha}_{1'}
\wedge {\mathcal
Y}^{\beta}_{1'}+\Psi^A\P_{A1'}\wedge\Y^0_{1'}\\[8pt]
&+&\mu^a\P_{a1'}\wedge\Y^0_{1'}+\Big[\rho_\alpha^i\Big(y_A\displaystyle\frac{\partial\Psi^A}{\partial
x^i}-\frac{\partial\tilde{L}}{\partial x^i}\Big)+C_{\alpha 0
}^\gamma y_\gamma\Big]\Y^\alpha_{1'}\wedge\Y^0_{1'}.
\end{array}$$
Thus, we obtain that
\begin{equation}\label{loczeta1}
\begin{array}{rcl}
\zeta_{1}(x^{j}, y_{\beta}) &=& \Y_{01'}+\mu^{a}(x^{j}, y_{\beta})
{\mathcal Y}_{a1'} + \Psi^{A}(x^{j}, \mu^{a}(x^{j}, y_{\beta}))
{\mathcal
Y}_{A1'}\\[8pt]& -& \Big[y_\gamma\Big(C_{\alpha 0}^\gamma+
\Psi^{A}(x^{j}, \mu^{a}(x^{j}, y_{\beta})){C}^{\gamma}_{\alpha A}+\mu^{a}(x^{j}, y_{\beta}){C}_{\alpha a}^{\gamma}\Big)\\[8pt]
& +&\rho^{i}_{\alpha} \Big( y_{A} \displaystyle \frac{\partial
\Psi^{A}}{\partial x^{i}}_{|(x^{j}, \mu^{a}(x^{j}, y_{\beta}))} -
\frac{\partial \tilde{L}}{\partial x^{i}}_{|(x^{j}, \mu^{a}(x^{j},
y_{\beta}))}\Big) \Big]{\mathcal P}_{1'}^{\alpha}.
\end{array}
\end{equation}

Now, we will introduce an aff-Poisson bracket on the AV-bundle
determined by the constraint submanifolds $W_1$ and $W_1'$. For
this propose, we define the application $\mu_1:W_1\to W_1'$ given
by
$$\mu_1(\varphi,\a)=(\varphi-H_{W_1}(\varphi,\a)1_\A(x),\a),$$
for $(\varphi,\a)\in W_1\subseteq W_0=\dual\oplus_Q\M$, with
$\nu_1(\varphi,\a)=x\in Q$.

If $(x^i,y_0,y_\alpha)$ (respectively, $(x^i,y_\alpha)$) are local
coordinates on $W_1$ (respectively, $W_1'$) as in Remark
\ref{obs6.18}, we deduce that the local expression of
$\mu_1:W_1\to W_1'$ is
$$\mu_1(x^i,y_0,y_\alpha)=(x^i,y_\alpha).$$

Moreover, we have the following result.

\begin{theorem}\label{th6.21} If $(L,\M)$ is a regular vakonomic system on
 the Lie affgebroid $\tau_\A:\A\to Q$, then
$\mu_1:W_1\to W_1'$ is a AV-bundle which admits an aff-Poisson
structure.
\end{theorem}
\begin{proof} It is easy to prove that $\mu_1:W_1\to W_1'$ is
an AV-bundle (see Section \ref{sec2.5}). In
fact, if $w=(\varphi,\a)\in (W_1)_x$, with $x\in Q$, and $t\in\R$
then
$$w+t=(\varphi,\a)+t=(\varphi+t1_\A(x),\a).$$

To define an aff-Poisson bracket on $\mu_1$ we will introduce a
Poisson bracket on $W_1$ which is invariant with respect to
$X_{W_1}$. Here, $X_{W_1}$ is the infinitesimal generator of the
principle action of $\R$ on $W_1$.

Consider the prolongation $\T (\pi_1)_{|W_1}:\T^\bidual
W_1\to\T^\bidual\dual$ of the restriction $(\pi_1)_{|W_1}:W_1\to
\A^+$ to $W_1$ of the application $\pi_1=\pr_{1|W_0}:W_0\to\dual$.
It is clear that $(\T (\pi_1)_{|W_1},(\pi_1)_{|W_1})$ is a Lie
algebroid morphism and, therefore, we can introduce the 2-section
$\Omega_{W_1}\in\Gamma(\wedge^2(\tau_\bidual^{\nu_1})^*)$ defined
by
$$\Omega_{W_1}=(\T (\pi_1)_{|W_1},(\pi_1)_{|W_1})^*\Omega_\bidual,$$
$\Omega_\bidual$ being the canonical symplectic 2-section on
$\T^\bidual\dual$. Obviously $d^{\T^\bidual W_1}\Omega_{W_1}=0$.

If $(x^i,y_0,y_\alpha)$ are local coordinates on $W_1$ as in
Remark \ref{obs6.18}, we can consider the local basis of sections
$\{\Y_{01},\Y_{\alpha 1}, \P^0_1,\P^\alpha_1\}$ of $\T^\bidual
W_1\to W_1$ given by
$$\begin{array}{rclcrcl}
{\mathcal Y}_{0 1}&=& \Big({\mathcal Y}_{0} + \rho^{i}_{0}
\displaystyle \frac{\partial \mu^{a}}{\partial x^{i}} {\mathcal
V}_{a} \Big)_{|W_{1}},&\;\;& {\mathcal Y}_{\alpha 1}& =&
\Big({\mathcal Y}_{\alpha} + \rho^{i}_{\alpha} \displaystyle
\frac{\partial \mu^{a}}{\partial x^{i}} {\mathcal V}_{a}
\Big)_{|W_{1}},\\[10pt]
{\mathcal P}_{1}^{0}& = &({\mathcal P}^{0})_{|W_{1}},&\;\;&
{\mathcal P}_{1}^{\alpha}& =& \Big({\mathcal P}^{\alpha} +
\displaystyle \frac{\partial \mu^{a}}{\partial y_{\alpha}}
{\mathcal V}_{a}\Big)_{|W_{1}}.
\end{array}
$$

If $\{\Y^{0}_1,\Y^{\alpha}_1, \P_{01},\P_{\alpha 1}\}$ is the dual
basis of $\{\Y_{01},\Y_{\alpha 1}, \P^0_1,\P^\alpha_1\}$, we
obtain that
$$\Omega_{W_1}=\Y^0_1\wedge\P_{01}+\Y^\alpha_1\wedge\P_{\alpha1}+C_{0\alpha}^\gamma
y_\gamma\Y^0_1\wedge\Y^\alpha_1+\displaystyle\frac{1}{2}C_{\alpha\beta}^\gamma
y_\gamma\Y^\alpha_1\wedge\Y^\beta_1.$$

Thus, we deduce that $\Omega_{W_1}$ is a symplectic section of
$\T^\bidual W_1\to W_1$ and, therefore, it induces a Poisson
bracket on $W_1$ which is given by
$$\{F,G\}_{W_1}=\Omega_{W_1}({\mathcal H}_F^{\Omega_{W_1}},{\mathcal
H}_G^{\Omega_{W_1}})=\rho_\bidual^{\nu_1}({\mathcal
H}_G^{\Omega_{W_1}})(F), \; \; \mbox{ for } F, G \in
C^{\infty}(W_{1}),$$ where ${\mathcal H}_F^{\Omega_{W_1}}$ and
${\mathcal H}_G^{\Omega_{W_1}}$ are the Hamiltonian sections
associated with the functions $F$ and $G$, respectively, with
respect to the symplectic structure $\Omega_{W_1}$.

Moreover, the Poisson 2-vector $\Pi_{W_1}$ determinated by the
bracket $\{\cdot , \cdot \}_{W_1}$ is invariant with respect to
$X_{W_1}$. In fact, we have that
$X_{W_1}=\displaystyle\frac{\partial}{\partial y_0}$ and
\begin{equation}\label{PoissW1}
\begin{array}{rcl}
\Pi_{W_1}&=&\rho_0^i\displaystyle\frac{\partial}{\partial
x^i}\wedge\frac{\partial}{\partial
y_0}+\rho_\alpha^i\displaystyle\frac{\partial}{\partial
x^i}\wedge\frac{\partial}{\partial y_\alpha}\\[10pt]
&&-C_{0\alpha}^\gamma
y_\gamma\displaystyle\frac{\partial}{\partial
y_0}\wedge\frac{\partial}{\partial
y_\alpha}-\displaystyle\frac{1}{2}C_{\alpha\beta}^\gamma
y_\gamma\frac{\partial}{\partial
y_\alpha}\wedge\frac{\partial}{\partial y_\beta}.
\end{array}\end{equation}

Thus, we conclude that $\mu_1:W_1\to W_1'$ admits an aff-Poisson
structure which we denote by $\{\cdot , \cdot
\}_{vak}:\Gamma(\mu_1) \times \Gamma(\mu_1) \to C^\infty(W_1').$
This structure is characterized by the following condition
$$\{h_1',h_1''\}_{vak}\comp\mu_1=\{F_{h_1'},F_{h_1''}\}_{W_1},\;\mbox{
for }\;h_1',h_1''\in\Gamma(\mu_1),$$ $F_{h_1'},F_{h_1''}$ being
the real functions on $W_1$ associated with the sections
$h_1',h_1''$ (as we know, $X_{W_1}(F_{h_1'})=
X_{W_1}(F_{h_1''})=-1$).

\end{proof}

The aff-Poisson bracket on the AV-bundle $\mu_1:W_1\to W_1'$,
$$\{\cdot , \cdot \}_{vak}:\Gamma(\mu_1) \times \Gamma(\mu_1) \to
C^\infty(W_1'),$$ is called {\it the vakonomic bracket associated
with the regular system} $(L,\M)$.

On the other hand, note that the restriction $H_{W_1}$ to $W_1$ of
the Pontryagin Hamiltonian $H_{W_0}$ verifies that
$X_{W_1}(-H_{W_1})=-1.$ Therefore, there exists
$h_1\in\Gamma(\mu_1)$ such that $F_{h_1}=-H_{W_1}.$ In fact, $h_1$
is the inclusion of $W_1'$ into $W_1$. Moreover, we have the
following result.

\begin{theorem} If $F_1'\in C^\infty(W_1')$ then the temporal evolution of $F_1'$, $\dot{F}_1'$, is given by
$$\dot{F}_1'=\{h_1,F_1'\}_{vak}^{al},$$
$\{\cdot,\cdot\}_{vak}^{al}$ being the affine-linear part of the
bi-affine bracket $\{\cdot,\cdot\}_{vak}$. In other words, the
Hamiltonian vector field associated with $h_1$ with respect to the
vakonomic bracket coincides with $\rho_\bidual^{\nu_1'}(\zeta_1)$.
\end{theorem}
\begin{proof} We know that the Hamiltonian vector field $\{h_1,\cdot\}_{vak}^{al}$ of $h_1$ with respect to
the vakonomic bracket is given by
\begin{equation}\label{cHafg}
\{h_1,\cdot\}_{vak}^{al}(\varphi)\comp \mu_1=\{
F_{h_1},\varphi\comp\mu_1\}_{W_1},\;\mbox{ for }\;\varphi\in
C^\infty(W_1').
\end{equation}
Then, from (\ref{loczeta1}), (\ref{PoissW1}), (\ref{cHafg}) and
Remark \ref{obs6.18}, we deduce that this vector field is just
$\rho_\bidual^{\nu_1'}(\zeta_1)$ (see (\ref{loczeta1})).
\end{proof}

Next, let $h_1',h_1''$ be two sections of $\mu_1:W_1\to W_1'$ and
suppose that
$$h_1'(x^i,y_\alpha)=(x^i,-H_1'(x^j,y_\beta),y_\alpha)\;\mbox{ and
}\;h_1''(x^i,y_\alpha)=(x^i,-H_1''(x^j,y_\beta),y_\alpha).$$ Then,
using (\ref{PoissW1}), we have that
\begin{equation}\label{corvakloc}
\begin{array}{rcl} \{
h_1',h_1''\}_{vak}&=&\rho_0^i\displaystyle\frac{\partial(H_1'-H_1'')}{\partial
x^i}+\rho_\alpha^i\Big(\frac{\partial H_1'}{\partial
x^i}\frac{\partial H_1''}{\partial y_\alpha}-\frac{\partial
H_1'}{\partial y_\alpha}\frac{\partial H_1''}{\partial x^i}\Big)\\[12pt]
&&-C_{\alpha 0}^\gamma y_\gamma\displaystyle\frac{\partial
(H_1'-H_1'')}{\partial y_\alpha}-C_{\alpha\beta}^\gamma
y_\gamma\frac{\partial H_1'}{\partial y_\alpha}\frac{\partial
H_1''}{\partial y_\beta}.
\end{array}\end{equation}

Since $\A$ is a Lie affgebroid, it follows that $\A^+$ is the
total space of an AV-bundle over $V^*$ with projection
$\mu:\A^+\to V^*$ and, moreover, the linear Poisson structure
$\Pi_{\A^+}$ on $\A^+$ (induced by the Lie algebroid structure of
$\bidual$) defines an aff-Poisson bracket
$\{\cdot,\cdot\}:\Gamma(\mu)\times\Gamma(\mu)\to C^\infty(V^*)$ on
the AV-bundle $\mu:\dual\to V^*$ (see Theorem \ref{teor2.2}).

On the other hand, we may consider the applications
$(\pi_1)_{|W_1}:W_1\kern-1pt\to\dual$ and
$\mu\comp(\pi_{1})_{|W_{1}'}: W_{1}' \to V^{*}$ and it is clear
that $\mu\comp(\pi_1)_{|W_1}=\mu\comp(\pi_1)_{|W_1'}\comp\mu_1$.

%

In fact, using (\ref{PoissW1}) and Remark \ref{obs6.18}, we can
prove the following result.

\begin{corollary}\label{cor6.22}
If $(L, \M)$ is a regular vakonomic system on $\A$, then the pair
$((\pi_1)_{|W_1},\mu\comp(\pi_{1})_{|W_{1}'})$ is a local
aff-Poisson isomorphism of AV-bundles, that is:
\begin{enumerate}
\item[i)] $\mu\comp(\pi_1)_{|W_1'}:W_1'\to V^*$ is a local diffeomorphism;
\item[ii)] The restriction of $(\pi_1)_{|W_1}$ to each fibre of $\mu_1:W_1\to
W_1'$ is an affine isomorphism over the corresponding fibre of
$\mu:\dual\to V^*$ and
\item[iii)] If $h_1',h_1''\in\Gamma(\mu_1)$ and $h',h''\in\Gamma(\mu)$ satisfy that
$$(\pi_1)_{|W_1}\comp h_1'=h'\comp\mu\comp(\pi_1)_{|W_1'},\;\;\;\;(\pi_1)_{|W_1}\comp
h_1''=h''\comp\mu\comp(\pi_1)_{|W_1'},$$ then
$\{h_1',h_1''\}_{vak}=\{h',h''\}\comp\mu\comp(\pi_1)_{|W_1'}.$
\end{enumerate}

\end{corollary}



\begin{remark}{\rm If $\mu\comp(\pi_1)_{|W_1'}:W_1'\to V^*$ is a
global diffeomorphism then we can define the section
$h\in\Gamma(\mu)$ given by
$$h=(\pi_1)_{|W_1}\comp h_1\comp(\mu\comp(\pi_1)_{|W_1'})^{-1}$$
and it is clear that $(\pi_1)_{|W_1}\comp h_1=h\comp
\mu\comp(\pi_1)_{|W_1'}$. This implies that the Hamiltonian vector
fields of $h_1$ and $h$ are
$((\pi_1)_{|W_1},\mu\comp(\pi_1)_{|W_1'})$-related. Therefore, if
$\gamma_1':I\to W_1'$ is a solution of the vakonomic equations for
the system $(L,\M)$, then
$\mu\comp(\pi_1)_{|W_1'}\comp\gamma_1':I\to V^*$ is a solution of
the Hamilton equations for $h$. Conversely, if $\gamma:I\to V^*$
is a solution of the Hamilton equations for $h$ then
$(\mu\comp(\pi_1)_{|W_1'})^{-1}\comp\gamma:I\to W_1'$ is a
solution of the vakonomic equations for the system $(L,\M)$.
}\end{remark}

\begin{remark}{\rm If $(L, {\mathcal M})$ is a vakonomic system on
a Lie affgebroid ${\mathcal A}$ which is not regular then one may
apply a constraint algorithm in order to obtain solutions of the
vakonomic equations in a suitable Lie subalgebroid of ${\mathcal
T}^{\widetilde{\mathcal A}}W_{1} \to W_{1}$.  }
\end{remark}

\subsection{Mechanical systems subject to affine
constraints on Lie affgebroids}

The Lagrangian function $L: {\mathcal A} \to \R$ of a mechanical
system on a Lie affgebroid $\tau_{\mathcal A}: {\mathcal A} \to Q$
is of the form
$$L(\a)=\displaystyle\frac{1}{2}{\mathcal
G}_{\tau_\A(\a)}(i_\A(\a),i_\A(\a))-{\mathcal
V}(\tau_\A(\a)),\;\mbox{ for all }\a\in\A,$$ where ${\mathcal G}$
is a bundle metric on $\bidual$ and $\V$ a function on $Q$. We
also denote by ${\mathcal G}$ the bundle metric induced on $\dual$
and we suppose that the 1-cocycle $1_\A$ has constant norm equal
to 1. Moreover, using the metric, we can identify $\A$ and $V$. In
fact, we have the affine bundle morphism ${\mathcal I}:\A\to V$
given by
$${\mathcal I}(\a)=i_\A(\a)-1_\A^{\mathcal G}(\tau_\A(\a)),\;\mbox{
for all }\;\a\in\A,$$
where $1_\A^{\mathcal G}\in\Gamma(\tau_\bidual)$ is defined by
$\varphi(1_\A^{\mathcal G}(x))={\mathcal G}_x(1_\A(x),\varphi)$,
for all $\varphi\in\A^+_x$. Then, if $\bar{\mathcal G}$ is the
restriction to $V$ of ${\mathcal G}$, the Lagrangian $L$ may be
written as follows
$$L(\a)=\displaystyle\frac{1}{2}\bar{\mathcal
G}_{\tau_\A(\a)}({\mathcal I}(\a),{\mathcal I}(\a))-\bar{\mathcal
V}(\tau_\A(\a)),\;\mbox{ for all }\a\in\A,$$ where $\bar{\mathcal
V}(x)={\mathcal V}(x)-\frac{1}{2}$, for all $x\in Q$.

Now, let $(x^i)$ be local coordinates on an open subset $U$ of
$Q$. Then, since the section $1_{\mathcal A}$ has constant norm
equal to $1$, we can consider an orthonormal basis of sections of
the vector bundle $\tau_{{\mathcal A}^{+}}^{-1}(U) \to U$ of the
form $\{e^0 = 1_{\mathcal A}, e^{\alpha}\}$. Thus, its dual basis
$\{e_{0}, e_{\alpha}\}$ is an orthonormal local basis of sections
of $\widetilde{\mathcal A}$. Moreover, if $(x^i,y^0,y^\alpha)$ are
the corresponding local coordinates on $\bidual$, the local
expression of the Lagrangian function is
$$L(x^i,y^\alpha)=\displaystyle\frac{1}{2}(y^\alpha)^2-\bar{\mathcal
V}(x^i)$$ and the Euler-Lagrange equations (that is, the vakonomic
equations for the system $(L, {\mathcal A})$) reduce to
$$\displaystyle\frac{dx^i}{dt}=\rho_0^i+\rho^i_\alpha
y^\alpha,\;\;\;\frac{dy^\alpha}{dt}=-\rho_\alpha^i\frac{\partial{\mathcal
V}}{\partial x^i}-(C_{\alpha 0}^\gamma+C_{\alpha\beta}^\gamma
y^\beta)y^\gamma.$$

Next, suppose that the constraint submanifold $\M$ of the
vakonomic system is an affine subbundle $\B$ of ${\mathcal A}$,
that is, we have an affine bundle $\B$ over $Q$ with associated
vector bundle $\tau_{U_\B}:U_\B\to Q$ and the corresponding
inclusions $i_\B:\B\to\A$ and $i_{U_\B}:U_\B\to V$. Furthermore,
assume that $1_{\mathcal A}^{\mathcal G} \in \Gamma(\tau_{B})$.
Then, we can choose an special coordinate system adapted to the
structure of the problem as follows. In fact, we consider local
coordinates $(x^i)$ on an open subset $U$ of $Q$ and an
orthonormal local basis of $\Gamma(\tau_{V})$, $\{e_{A}, e_{a}\}$,
adapted to the decomposition $V=U_\B^{\perp,{\bar{\mathcal
G}}}\oplus U_\B$, $U_\B^{\perp,{\bar{\mathcal G}}}$ being the
orthogonal subbundle to $U_{\B}$ with respect to the bundle metric
$\bar{\mathcal G}$. Thus, we deduce that $\{1_{\mathcal
A}^{\mathcal G}= e_{0}, e_{A}, e_{a}\}$ is an orthonormal local
basis of $\Gamma(\tau_{\widetilde{\mathcal A}})$ adapted to the
affine subbundle $\B$. Denote by $(x^i, y^0, y^A, y^a)$ the
corresponding local coordinates on $\widetilde{\mathcal A}$ and by
$(x^i, y_0, y_A, y_a)$ the dual local coordinates on ${\mathcal
A}^+$. Note that the equations which define to $\B$ as an affine
subbundle of ${\mathcal A}$ are $y^A = 0$. Therefore, the
vakonomic system $(L, \B)$ is regular and the local expression of
the vakonomic equations is
$$\left\{ \begin{array}{l}\displaystyle\frac{dx^i}{dt}=\rho_0^i+\rho_a^iy^a,\\[8pt]
\displaystyle\frac{dy_\alpha}{dt}=-\rho_\alpha^i\frac{\partial{\mathcal
V} }{\partial x^i}-(C_{\alpha0}^\gamma+ C_{\alpha a}^\gamma y_a)y_\gamma,\\[6pt]
y^0=1,\;y^A=0,\; y^a = y_a,\; y_{0} = \displaystyle \frac{1}{2}
(y_a)^2 + {\mathcal V}(x^i) . \end{array}\right. $$


\subsection{The variational point of view}

Let $\tau_\A :\A\to Q$ be a Lie affgebroid modelled on the Lie
algebroid $\tau_V :V\to Q$ and $L:\A\to \R$ be a Lagrangian
function on $\A$. Next, we will show how to obtain the
Euler-Lagrange equations on the Lie affgebroid $\A$ from a
variational point of view.

We define the set of $\A$-paths as follows  $${\mathcal
Adm}([t_0,t_1],\A)=\{ \a:[t_0,t_1]\to \A \, | \, \rho_\A\comp
\a=\displaystyle\frac{d}{dt}(\tau_\A\comp \a) \},$$ that is, as
the set of admissible curves in $\A$. Then, for two fixed points
$x,y\in Q$, denote by ${\mathcal Adm}([t_0,t_1],\A)_{x}^{y}$ the
set of $\A$-paths with fixed base
endpoints equal to $x$ and $y$. 

Now, if $i_V:V\to\bidual$ is the canonical inclusion, we consider
as infinitesimal variations the complete lifts of sections of
$\tau_V:V\to Q$ which vanish at the points $x$ and $y$, that is,
\[ \{(i_V\comp \bar{X})^{c}_{|\A}\, |
\,\bar{X}\in\Gamma(\tau_V),\,\bar{X}(x)=0 \mbox{ and }
\bar{X}(y)=0\}.
\]
Note that if $\{e_0,e_\alpha\}$ is a local basis of
$\Gamma(\tau_\bidual)$ and $\bar{X}\in\Gamma(\tau_V)$ is locally
given by $\bar{X}=\bar{X}^\alpha e_\alpha$, then $(i_V\comp
\bar{X})^{c}_{|\A}$ is the vector field on $\A$ given by
$$(i_V\comp
\bar{X})^{c}_{|\A}=\bar{X}^c_i\displaystyle\frac{\partial}{\partial
x^i}+\bar{X}^c_\alpha\frac{\partial}{\partial y^\alpha},$$ where
$\bar{X}^c_i=\bar{X}^\alpha\rho_\alpha^i,\;\;\bar{X}^c_\alpha=\frac{\partial\bar{X}^\alpha}{\partial
x^i}(\rho_0^i+y^\beta\rho_\beta^i)-\bar{X}^\gamma(C_{\gamma
0}^\alpha+y^\beta C_{\gamma\beta}^\alpha),\,\mbox{ for all
}\,i\,\mbox{ and }\,\alpha.$

On the other hand, we introduce the action functional $\delta
S:{\mathcal Adm}([t_0,t_1],\A)\to \R$ defined by
\[
\delta S(\a)=\int_{t_0}^{t_1} L(\a(t))dt.
\]
With this definition it is not difficult to prove that the
critical points of $\delta S$ on ${\mathcal
Adm}([t_0,t_1],\A)_{x}^{y}$ are the curves $\a \in {\mathcal
Adm}([t_0,t_1],\A)_x^y$ which satisfy the Euler-Lagrange equations
(that is, the vakonomic equations for the system $(L, {\mathcal
M})$, with ${\mathcal M} = {\mathcal A}$).

Now, let $(L, \M)$ be a vakonomic system on the Lie affgebroid
$\tau_\A: \A \to Q$. Denote by ${\mathcal Adm}([t_0,t_1],\M)_x^y$
the set of $\A$-paths on $\M$ with fixed base endpoints equal to
$x$ and $y$, respectively, that is,
\[
{\mathcal Adm}([t_0,t_1], \M)_x^y =\big \{ \a\in {\mathcal
Adm}([t_0,t_1],\A)_x^y \,|\, \a(t)\in \M, \,\forall\, t\in
[t_0,t_1] \big \}.
\]
In this case, we are going to consider infinitesimal variations
(that is, complete lifts $(i_V\comp\bar{X})^{c}_{|\A}$, with
$\bar{X}\in\Gamma(\tau_V)$) tangent to the constraint submanifold
$\M$ and we assume that there exist enough infinitesimal
variations of this kind (that is, we are studying the so-called
normal solutions of the vakonomic problem). If $\M$ is locally
given by the equations $y^A - \Psi ^A (x^i, y^a)=0$, for
$A=1,\dots,\bar{m}$, we deduce that the infinitesimal variations
must satisfy
\[
 (i_V\comp\bar{X})^{c}_{|\A} (y^A - \Psi ^A (x^i, y^a))=0,\quad \bar{X}(x)=0,\quad
 \bar{X}(y)=0.
\]
Note that if $\a\in {\mathcal Adm}([t_0,t_1],\M)_x^y$ then
\[
(i_V\comp\bar{X})^{c}_{|\A} (y^A - \Psi ^A (x^i, y^a))\comp \a=0
\]
if and only if
\begin{equation}\label{varinfafg}\begin{array}{rcl}
\displaystyle\frac{d\bar{X}^A}{dt}&=&\rho^i_\alpha\bar{X}^\alpha\displaystyle\frac{\partial
\Psi ^A}{\partial x^i}+\frac{d\bar{X}^a}{dt}\frac{\partial \Psi
^A}{\partial y^a}-(C_{\gamma
0}^a+{C}^a_{\gamma\beta}y^\beta)\bar{X} ^\gamma\frac{\partial \Psi
^A}{\partial y^a}\\[10pt]&&+(C_{\gamma
0}^A+{C}^A_{\gamma\beta}y^\beta)\bar{X}^\gamma.
\end{array}\end{equation}
Thus, if we consider our infinitesimal variations then
\[
\begin{array}{rcl}
\displaystyle \frac{d}{ds}_{|s=0} \int_{t_0}^{t_1} \kern-5ptL
(\a_s(t))dt &\kern-5pt=&\kern-7pt\displaystyle
\int_{t_0}^{t_1}\kern-5pt\Big( \frac{\partial L}{\partial
x^i}\bar{X}^{c}_i + \frac{\partial L}{\partial y^A}\frac{\partial
\Psi ^A}{\partial x^i} \bar{X}^{c}_i  + \frac{\partial L}{\partial
y^a}\bar{X}^{c}_a+ \displaystyle\frac{\partial L}{\partial
y^A}\frac{\partial
\Psi ^A}{\partial y^a} \bar{X}^{c}_a \Big) dt\\[14pt]

 &\kern-5pt=&
\kern-7pt\displaystyle \int_{t_0}^{t_1} \Big(\frac{\partial
\tilde{L}}{\partial x^i}\rho ^i_A \bar{X} ^A + \frac{\partial
\tilde{L}}{\partial x^i}\rho ^i_a \bar{X} ^a + \frac{\partial
\tilde{L}}{\partial y^a}\bar{X}^{c}_a \Big)dt
\end{array}
\]
Now, let $y_A$ be the solution of the differential equations
\[
\dot{y}_{A}=\Big ( \frac{\partial \tilde{L}}{\partial x^i}-y_B
\frac{\partial \Psi ^B}{\partial x^i} \Big ) \rho ^i_{A}
-y_\gamma(C_{A0}^\gamma+\Psi^B{C}^\gamma_{A B}+
y^a{C}^\gamma_{Aa}),
\]
where
\begin{equation}\label{momenta-eq-afg}
y_a=\frac{\partial \tilde{L}}{\partial y^a}-y_A \frac{\partial
\Psi ^A}{\partial y^a}.
\end{equation}
Then, from (\ref{varinfafg}), we have that
\[
\begin{array}{l}
\displaystyle \frac{d}{dt}( y_A \bar{X}^A)= \displaystyle
\dot{y}_A \bar{X}^A + y_A
\dot{\bar{X}}^A=\displaystyle\frac{\partial \tilde{L}}{\partial
x^i}\rho ^i_{A}\bar{X}^A +\displaystyle y_A
\rho^i_a\bar{X}^a\frac{\partial \Psi ^A}{\partial
x^i}+y_A\frac{d\bar{X}^a}{dt}\frac{\partial \Psi
^A}{\partial y^a}\\[10pt]\hfill-y_A(C_{\gamma
0}^a+\Psi^B C^a_{\gamma B} +y^b{C}^a_{\gamma
b})\bar{X}^\gamma\displaystyle\frac{\partial \Psi ^A}{\partial
y^a} +y_A(C_{b0}^A+\Psi^B C_{bB}^A\\[10pt]
\hfill+y^c C_{b c}^A)\bar{X}^b -
y_b(C_{A0}^b+\Psi^BC_{AB}^b+y^aC_{Aa}^b) \bar{X} ^A.
\end{array}
\]
Using this equality, we deduce that
\[
\begin{array}{l}
\displaystyle \frac{d}{ds}_{|s=0} \int_{t_0}^{t_1} L (\a_s(t))dt
=\displaystyle
 \int_{t_0}^{t_1} \Big (
 \frac{d}{dt}(y_A\bar{X}^A)+\frac{\partial\tilde{L}}{\partial x^i}\rho_a^i\bar{X}^a+
 \frac{\partial\tilde{L}}{\partial y^a}\bar{X}^c_a\\[10pt]
\hfill-y_A\rho_a^i\bar{X}^a\displaystyle\frac{\partial\Psi^A}{\partial
 x^i}-y_A\frac{d\bar{X}^a}{dt}\frac{\partial\Psi^A}{\partial
 y^a}+y_A(C_{\gamma 0}^a+\Psi^B C_{\gamma B}^a+y^b C_{\gamma
 b}^a)\bar{X}^\gamma\displaystyle\frac{\partial\Psi^A}{\partial
 y^a}\\[10pt]
\hfill
+y_b(C_{A0}^b+\Psi^BC_{AB}^b+y^aC_{Aa}^b)\bar{X}^A-y_A(C_{b0}^A+\Psi^BC_{bB}^A+y^cC_{bc}^A)\bar{X}^b
\Big)dt .
\end{array}
\]
Finally, using (\ref{momenta-eq-afg}) and the fact that
\[
\bar{X}^{c}_a=\frac{d\bar{X} ^a}{dt}- (C_{\gamma
0}^a+\Psi^AC_{\gamma A}^a+y^bC_{\gamma b}^a )\bar{X}^\gamma,
\]
we obtain that
\[
\begin{array}{rcl}
\displaystyle \frac{d}{ds}_{|s=0} \int_{t_0}^{t_1} L (\a_s(t))dt
&=& \displaystyle
 \int_{t_0}^{t_1} \Big [
\Big ( \frac{\partial \tilde{L}}{\partial x^i}-y_A \frac{\partial
\Psi ^A}{\partial x^i} \Big )\rho ^i_{a} -\frac{d}{dt} \Big (
\frac{\partial \tilde{L}}{\partial y^a} -y_A \frac{\partial \Psi
^A}{\partial y^a} \Big )\\[12pt]&& -y_\gamma(C_{a0}^\gamma+\Psi^BC_{aB}^\gamma+y^b C_{ab}^\gamma) \Big]\bar{X}^a dt .
\end{array}
\]
Since the variations $\bar{X}^a$ are free, we conclude that the
equations are

$$\left \{
\begin{array}{l}
\displaystyle \dot{x}^i=\rho_0^i+\Psi ^A \rho ^i_A+y^a \rho ^i_a ,\\[5pt]
\displaystyle \dot{y}_{A}=\Big ( \frac{\partial
\tilde{L}}{\partial x^i}-y_B \frac{\partial \Psi ^B}{\partial x^i}
\Big ) \rho ^i_A -y_\gamma(C^{\gamma}_{A0}+\Psi ^B C^\gamma_{AB}+y^aC^\gamma_{Aa}) ,\\[10pt]
\displaystyle \frac{d}{dt}\left( \frac{\partial
\tilde{L}}{\partial y^a}-y_A \frac{\partial \Psi ^A}{\partial y^a}
\right)= \Big ( \frac{\partial \tilde{L}}{\partial x^i}-y_A
\frac{\partial \Psi ^A}{\partial x^i} \Big ) \rho ^i_{a}\\[12pt]
\hspace{4.5cm}-y_\gamma( C^{\gamma}_{a0}+\Psi ^B
C^\gamma_{aB}+y^bC^\gamma_{ab}),
\end{array}
\right .$$ for all $1\leq i\leq m$, $1\leq A\leq\bar{m}$ and
$\bar{m}+1\leq a\leq n$, with $y_a=\displaystyle\frac{\partial
\tilde{L}}{\partial y^a}-y_A \frac{\partial \Psi^A}{\partial
y^a}$, that is, the vakonomic equations for the system $(L,\M)$ on
$\tau_\A :\A\to Q$ (see (\ref{eqvakafg})).

\section{Examples}

\subsection{Skinner-Rusk formalism on Lie affgebroids}
Consider on a Lie affgebroid $\tau_{\A}:\A\rightarrow Q$ a
vakonomic system $(L,\M)$ with $\M=\A$, that is, a free system. In
this case, $W_0=\dual\oplus_Q \A$ and the Pontryagin Hamiltonian
$H_{W_0}:\dual\oplus_Q \A\to \R$ is defined by
$H_{W_0}(\varphi,\a)=\varphi(\a)-L(\a)$. Moreover, the
precosymplectic structure $(\Omega _{W_0},\eta)$ on $\T^\bidual
W_0$ is given by
$$\Omega _{W_0}=({\mathcal T}\mathrm{pr}_1, \mathrm{pr}_1)^*\Omega
_\bidual+d^{\T^\bidual W_0}H_{W_0}\wedge\eta\;\mbox{ and
}\;\eta=(\tilde{\mathrm{pr}}_1,\nu)^*1_\A,$$ where
$\mathrm{pr}_1:\dual\oplus_Q\A\to\dual$ is the canonical
projection on the first factor,
$\T\mathrm{pr}_1:\T^\bidual(\dual\oplus_Q\A)\to\T^\bidual\dual$ is
its prolongation and
$\tilde{\mathrm{pr}}_1:\T^\bidual(\dual\oplus_Q\A)\to\bidual$ is
the restriction of the projection $\bidual\times
T(\dual\oplus_Q\A)\to\bidual$ (on the first factor) to the
prolongation $\T^\bidual(\dual\oplus_Q\A)$. In local coordinates,
we have that
$$H_{W_0}(x^i,y_0,y_\alpha,y^\alpha)=y_0+y_\alpha
y^\alpha-L(x^i,y^\alpha),$$
$$\begin{array}{rcl}
\Omega_{W_0}&=&\Y^\alpha\wedge {\mathcal
P}_\alpha+\Big(C_{\alpha0}^\gamma
y_\gamma-\rho_\alpha^i\displaystyle\frac{\partial {L}}{\partial
x^i}\Big)\Y^\alpha\wedge\Y^0+ \displaystyle\frac{1}{2}
{C}_{\alpha\beta}^\gamma y_\gamma \Y^\alpha\wedge
\Y^\beta\\[12pt]&&+\;y^\alpha
\P_\alpha\wedge\Y^0 + \Big(y_\alpha-\displaystyle\frac{\partial
L}{\partial y^\alpha}\Big)\V^\alpha\wedge\Y^0\end{array}$$ and
$$\eta=\Y^0.$$ Then, the submanifold
$W'_1\subseteq\dual\oplus_Q\A$ is locally characterized by
$$y_{0} = L(x^i, y_{\alpha}) - y_{\alpha}y^{\alpha}, \; \; \; \;
y_\alpha-\displaystyle\frac{\partial L}{\partial y^\alpha}=0$$
and the vakonomic equations reduce to
$$\left \{\begin{array}{rcl}
\dot{x}^i&=&\rho_0^i+y^\alpha \rho_\alpha^i,\\[6pt]
\displaystyle\frac{d}{dt}\Big(\frac{\partial L}{\partial
y^\alpha}\Big)&=&\rho_{\alpha}^i\displaystyle\frac{\partial
L}{\partial x^i}-(C_{\alpha 0
}^{\gamma}+C_{\alpha\beta}^{\gamma}y^{\beta})\displaystyle\frac{\partial
L}{\partial y^{\gamma}}. \end{array}\right .$$

Thus, if $(\pr_2)_{|W_1}:W_1\to\A$ is the restriction to $W_1$ of
the canonical projection on the second factor and $\gamma_1:I\to
W_1$ is a solution of the vakonomic equations, then
$(\pr_2)_{|W_1}\comp\gamma_1$ is a solution of the Euler-Lagrange
equations for $L$.

Note that in the standard case, that is, if $\A=J^1\tau$, this
procedure is the Skinner-Rusk formulation for time-dependent
mechanics (see \cite{BLEE,CoMaC}).


\subsection{The 1-jet bundle of local sections of a fibration}
Let $\tau:Q\to\R$ be a fibration and $\tau_{1,0}:J^1\tau\to Q$ be
the associated Lie affgebroid modelled on the vector bundle
$\pi=(\pi_Q)_{|V\tau}:V\tau\to Q$ (see Section \ref{secaff}). If
$(t,q^i)$ are local fibred coordinates on $Q$ then $\{
\frac{\partial}{\partial t}, \frac{\partial}{\partial q^i} \}$ is
a local basis of sections of $\pi_Q:TQ\to Q$. Denote by
$(t,q^i,\dot{t},\dot{q}^i)$ the corresponding local coordinates on
$TQ$. Then,  the (local) structure functions of $TQ$ with respect
to this local trivialization are given by
\begin{equation}\label{str-const-trivial}
\begin{array}{l}
C^k_{ij}=0 \mbox{ and } \rho ^i_j=\delta _{ij}, \mbox{ for }i,j, k
\in \{ 0,1,\ldots ,n\}.
\end{array}
\end{equation}
Moreover, $(t,q^i,\dot{q}^i)$ are the corresponding local
coordinates on $J^1\tau$.

Now, let $\M\subseteq J^1\tau$ be a constraint submanifold such
that $\tau_{1,0|\M}:\M\to Q$ is a surjective submersion and
$L:J^1\tau\to\R$ be a Lagrangian function. Suppose that the
constraint submanifold $\M$ is locally defined by the equations
$\dot{q}^A=\Psi^A(t,q^i,\dot{q}^a),$ where we use the following
notation $(t,q^i,\dot{q}^i)=(t,q^i,\dot{q}^A,\dot{q}^a)$.

Thus, if we apply the results of the Section \ref{sec4.1} to this
particular case, we recover some results obtained in \cite{BLEE}.
In particular, using (\ref{eqvakafg}) and
(\ref{str-const-trivial}), it follows that the vakonomic equations
reduce to
\[
\left \{
\begin{array}{l}\displaystyle \dot{p}_A =\frac{\partial \tilde{L}}{\partial
q^A }-p_B \frac{\partial \Psi ^B}{\partial
q^A},\\[10pt] \displaystyle  \frac{d}{dt}\Big ( \frac{\partial
\tilde{L}}{\partial
\dot{q}^a}-p_A\frac{\partial\Psi^A}{\partial\dot{q}^a}\Big )=
\frac{\partial \tilde{L}}{\partial q^a}- p_A \frac{\partial \Psi
^A}{\partial {q}^a},\\[10pt]
\displaystyle \dot{q}^A =\Psi^A(t,q^i,\dot{q}^a),
\end{array}
\right .
\]
where $(t,q^i,p,p_i)$ are the local coordinates on $T^*Q$ induced
by the local basis $\{dt,$ $dq^i\}$. If these equations are
written using the Lagrange multipliers (see (\ref{eqvakafg2}))
then they coincide with the equations obtained in \cite{ViB}.

On the other hand, if $(L,\M)$ is a regular vakonomic system on
$\tau_{1,0}:J^1\tau\to Q$, then the AV-bundle $\mu_1:W_1\to W_1'$
is locally defined by $\mu_1(t,q^i,p,p_i)=(t,q^i,p_i)$. Thus, if
$h_1',h_1'':W_1'\to W_1$ are two sections of $\mu_1:W_1\to W_1'$,
$$h_1'(t,q^i,p_i)=(t,q^i,-H_1'(t,q^j,p_j),p_i)\mbox{ and }h_1''(t,q^i,p_i)
=(t,q^i,-H_1''(t,q^j,p_j),p_i),$$
then, from (\ref{corvakloc}), we deduce that the vakonomic bracket
$\{\cdot , \cdot \}_{vak}:\Gamma(\mu_1) \times \Gamma(\mu_1) \to
C^\infty(W_1')$ is locally given by
$$\{h_1',h_1''\}_{vak}=\displaystyle\frac{\partial(H_1'-H_1'')}{\partial
t}+\frac{\partial H_1'}{\partial q^i}\frac{\partial
H_1''}{\partial p_i}-\frac{\partial H_1'}{\partial
p_i}\frac{\partial H_1''}{\partial q^i}.$$

\subsection{Optimal control systems as vakonomic systems on Lie affgebroids} Let
$\tau_\A:\A\to Q$ be a Lie affgebroid and $C$ a fibred manifold
over the state manifold $\pi :C\to Q$. We also consider a section
$\sigma :C\to \A$ along $\pi$ and an index function $l:C\to \R$.
The triple $(l, \pi, \sigma)$ is an optimal control system on the
Lie affgebroid ${\mathcal A}$.

One important case happens when the section $\sigma :C\to \A$
along $\pi$ is an embedding. In such a case, we have that the
image $\M=\sigma(C)$ is a submanifold of $\A$. Moreover, since
$\sigma: C\to \M$ is a diffeomorphism, we can define a Lagrangian
function $L: \M\to \R$ by $L=l\comp \sigma^{-1}$. Therefore, it is
equivalent to analyze the optimal control problem defined by $(l,
\pi, \sigma)$ (applying the Pontryagin maximum principle) that to
study the vakonomic problem on the Lie affgebroid $\tau_\A :\A\to
Q$ defined by $(L, \M)$.

In the general case (when $\sigma: C \to {\mathcal A}$ is not, in
general, an embedding), we consider the subset $\J^\A C$ of the
product manifold $\A\times TC$ defined by
$$\J^\A C=\{(\a,v)\in \A\times TC
\,|\,\rho_\A(\a)=(T\pi)(v)\}.$$ Next, we will show that $\J^\A C$
admits a Lie affgebroid structure. Let $\tau_\bidual ^\pi
:{\mathcal T}^\bidual C\to C$ be the prolongation of the bidual
Lie algebroid $\bidual$ of $\A$ over the fibration $\pi :C\to Q$.

On the other hand, let $\phi:\T^\bidual C\to\R$ be the section of
$(\tau_\bidual^\pi)^*:(\T^\bidual C)^*\to C$ defined by
$$\phi(\tilde{\a},X_p)=1_\A(\tilde{\a}),\,\mbox{ for }\,(\tilde{\a},X_p)\in\T^\bidual_p C,\,\mbox{ with }\,p\in C.$$
Note that if $\pr_1:{\mathcal  T}^{\widetilde{\A}}C\to
\widetilde{\A}$ is the canonical projection on the first factor
then $(\pr_1,\pi)$ is a morphism between the Lie algebroids
$\tau_{\widetilde{\A}}^{\pi}:{\mathcal
T}^{\widetilde{\A}}C\rightarrow C$ and
$\tau_{\widetilde{\A}}:\widetilde{\A}\to Q$ and, moreover, we have
that $(\pr_1,\pi)^*(1_{\A})=\phi.$ Since $1_{\A}$ is a $1$-cocycle
of $\tau_{\widetilde{\A}}:\widetilde{\A}\rightarrow Q$, we deduce
that $\phi$ is a $1$-cocycle of the Lie algebroid
$\tau_{\widetilde{\A}}^{\pi}:{\mathcal
T}^{\widetilde{\A}}C\rightarrow C$ and, using the fact that
$(1_\A)_{|\bidual_x}\neq 0$, for all $x\in Q$, we have that
$\phi_{|\T^\bidual_p C}\neq 0$, for all $p\in C$.

In addition, it follows that
$$\phi^{-1}\{1\}=\{(\tilde{\a},X_p)\in \T^\bidual_p C
\,|\,1_\A(\tilde{\a})=1\}=\J^\A C.$$

On the other hand, let $\tau_V^\pi:\T^V C\to C$ be the
prolongation of the Lie algebroid $(V,\lcf\cdot,\cdot\rcf_V,$
$\rho_V)$ over the fibration $\pi:C\to Q$. Then, it is easy to
prove that $\phi^{-1}\{0\}=\T^VC.$ Therefore, we conclude that
$\J^\A C$ is an affine bundle over $C$ with projection
$\tau_\A^\pi:\J^\A C\to C$ defined by $\tau_\A^\pi(a,v)=\pi_C(v),$
where $\pi_C:TC\to C$ is the canonical projection. Moreover, the
affine bundle $\tau_\A^\pi:\J^\A C\to C$ admits a Lie affgebroid
structure such that its bidual Lie algebroid is just $(\T^\bidual
C,\lcf\cdot,\cdot\rcf_\bidual^\pi,\rho_\bidual^\pi)$ and it is
modelled on the Lie algebroid $\tau_V^\pi:\T^V C\to C$ (see
Section \ref{secaff}).

Thus, we can consider the constraint submanifold $\M$ of the Lie
affgebroid $\J^\A C$ defined by
\[ \M =\bigcup_{p\in C} \{ (\a,X_p) \in {\mathcal J}^\A_pC \, | \, \sigma (p)=\a \}
\]
and the Lagrangian function $L :{\mathcal J}^\A C \to \R$ given by
$L = l\comp \tau_\A ^\pi$. Then, $(L,\M)$ is \textit{the vakonomic
system associated with the optimal control system}.

If $\A=J^1\tau$ is the 1-jet bundle of local sections of a
fibration $\tau:Q\to\R$, it is easy to prove that the prolongation
of $\widetilde{J^1\tau}\cong TQ$ over $\pi :C\to Q$ is just $TC$.
Thus, $\J^\A C = \{ X\in TC\,|\, dt(X)=1\} \cong J^1(\tau \comp
\pi)$, $t$ being the usual coordinate on $\R$. Under these
identifications, the constraint submanifold is
\[
\M = \{ X \in TC \, | \, T\pi (X)=\sigma (\pi _C(X) ) \}.
\]
Therefore, we recover the construction developed in \cite{BLEE}.

\begin{example}
{\rm We consider the following mechanical problem (see
\cite{BKMM,CLMM,IMMS,LM,NF}). A (homogeneous) sphere of radius
$r>0$, mass $m$ and inertia $mk^2$ about any axis rolls without
sliding on a horizontal table which rotates with time-dependent
angular velocity about a vertical axis through one of its points.
Apart from the constant gravitational force, no other external
forces are assumed to act on the sphere. The configuration space
of the sphere is $Q=\R^3\times SO(3)$ and the Lagrangian of the
system corresponds with the kinetic energy
\[
K(t,x,y;\dot{t}, \dot{x}, \dot{y}, \omega_{x}, \omega_{y},
\omega_{z})=\frac{1}{2}(m\dot{{x}}^2+m\dot{{y}}^2 +
mk^2(\omega_{x}^2 + \omega_{y}^2 + \omega_{z}^2)),
\]
where $(\omega_{x}, \omega_{y}, \omega_{z})$ are the components of
the angular velocity of the sphere.

Since the ball is rolling without sliding on a rotating table then
the system is subjected to the affine constraints:
\[
\dot{{x}}-r\omega_{y}=-\Omega(t) y, \; \; \;  \dot{{y}} +r
\omega_{x}=\Omega(t) x,
\]
where $\Omega(t)$ is the angular velocity of the table. Moreover,
it is clear that $Q=\R^3\times SO(3)$ is the total space of a
trivial principal $SO(3)$-bundle over $\R^3$ and the bundle
projection $\pi:Q\to \R^3$ is just the canonical projection on the
first factor. Therefore, we may consider the corresponding Atiyah
Lie algebroid $TQ/SO(3)$ over $\R^3$ (see \cite{LMM,Ma}).

Since the Atiyah Lie algebroid $TQ/SO(3)$ is isomorphic to the
product manifold $T\R^3\times {\mathfrak {so}}(3)\cong T\R^3\times
\R^3$, then a section of $TQ/SO(3)\cong T\R^3\times \R^3 \to \R^3$
is a pair $(X,u)$, where $X$ is a vector field on $\R^3$ and
$u:\R^3\to \R^3$ is a smooth map. Therefore, a global basis of
sections of $T\R^3\times \R^3\to \R^3$ is
\[\Big\{e_0=(\displaystyle\frac{\partial}{\partial {t}},0),\,
e_1=(\displaystyle\frac{\partial}{\partial {x}},0),\,
e_2=(\displaystyle\frac{\partial}{\partial y},0),\, e_3=(0,u_1),\,
e_4=(0,u_2),\, e_5=(0,u_3)\Big\},
\]
where $u_1, u_2, u_3:\R^3\to\R^3$ are the constant maps
$u_1(t,x,y)=(1,0,0)$, $u_2(t,x,y)=(0,1,0)$ and
$u_3(t,x,y)=(0,0,1)$.

The anchor map $\rho_{TQ/SO(3)}: T\R^3\times \R^3\to T\R^3$ is the
projection over the first factor and if $\lcf\cdot, \cdot
\rcf_{TQ/SO(3)}$ is the Lie bracket on the space
$\Gamma(\tau_{TQ/SO(3)})$ then the only non-zero fundamental Lie
brackets are
\[
\lcf e_3,e_4\rcf_{TQ/SO(3)}=e_5,\;\;\;\lcf
e_4,e_5\rcf_{TQ/SO(3)}=e_3,\;\;\; \lcf e_5,e_3\rcf_{TQ/SO(3)}=e_4.
\]

Moreover, $\phi: T\R^3\times\R^3\to \R$ given by $\phi(t, x, y;
\dot{t}, \dot{x}, \dot{y}, \omega_x, \omega_y, \omega_z)=\dot{t}$
is a 1-cocycle in the co\-rres\-ponding Lie algebroid cohomology
and, then, it induces a Lie affgebroid structure over ${\mathcal
A}=\phi^{-1}\{1\}\equiv \R\times T\R^2\times \R^3$. In addition,
the affine bundle $\tau_{\mathcal A}: \R\times T\R^2\times
\R^3\rightarrow \R^3$ is modelled on the vector bundle
$\tau_V:V=\phi^{-1}\{0\}\equiv \R\times T\R^2\times \R^3\to \R^3$
and its bidual Lie algebroid $\bidual$ is just the Atiyah Lie
algebroid $T\R^3\times\R^3$. Note that the Lie affgebroid
structure on ${\mathcal A}= \R\times T\R^2\times \R^3$ is a
special type of Lie affgebroid structure called Atiyah Lie
affgebroid structure (see Section 9.3.1 in \cite{IMPS} for a
general construction). Thus, $(t, x, y; \dot{x}, \dot{y},
\omega_x, \omega_y, \omega_z)$ may be considered as local
coordinates on ${\mathcal A}$ and $V$.

It is clear that the Lagrangian function and the nonholonomic
constraints are defined on the Atiyah Lie affgebroid ${\mathcal
A}\equiv \R\times T\R^2\times\R^3$ (since the system is
$SO(3)$-invariant). In fact, we have a nonholonomic system on the
Atiyah Lie affgebroid ${\mathcal A}\equiv \R\times
T\R^2\times\R^3$ (see \cite{IMMS} for more details).

After some computations the equations of motion for this
nonholonomic system may be written as follows
\begin{equation}\label{qweafg}
\dot{{x}}-r\omega_{y}=-\Omega(t)y , \; \; \; \dot{{y}} +r
\omega_{x}=\Omega(t)x, \; \; \; \omega_{z}=c,
\end{equation}
where $c$ is a constant, together with
$$\ddot{{x}}+\displaystyle\frac{k^2}{k^2+r^2}(\Omega'(t)y+\Omega(t)\dot{{y}})=0,
\;\;\;\ddot{{y}}-\frac{k^2}{k^2+r^2}(\Omega'(t)x+\Omega(t)\dot{{x}})=0.$$
Now, we pass to an optimization problem. Assume full control over
the motion of the center of the sphere and consider the cost
function
\[
L(t,x,y; \dot{{x}}, \dot{{y}}, \omega_{x}, \omega_{y},
\omega_{z})=\frac{1}{2}\left( (\dot{{x}})^2+(\dot{{y}})^2\right)
\]
and the following optimal control problem: {\sl Given two points
$q_0, q_1\in Q$, find an optimal control curve $(t,{x}(t),
{y}(t))$ on the reduced space that steer the system from $q_0$ and
$q_1$ and minimizes $\int_0^1 \frac{1}{2}\left(
(\dot{{x}})^2+(\dot{{y}})^2\right)\, dt$ subject to the
constraints defined by equations (\ref{qweafg}). }

Note that this problem is equivalent to the optimal control
problem defined by the section $\sigma: \R^3\times \R^2\to\R\times
T\R^2\times \R^3$ along $\R\times T\R^2\times \R^3\to \R^3$ given
by
\[
\sigma(t,x, y; u^1, u^2)=(t,x, y; u^1,
u^2,\frac{1}{r}(-{u^2}+\Omega(t) {x}),\frac{1}{r}(
{u^1}+\Omega(t)y),c)
\]
and the index function $l(t,x, y; u^1,
u^2)=\frac{1}{2}((u^1)^2+(u^2)^2)$. Since $\sigma$ is obviously an
embedding, we deduce the equivalence between both problems.

A necessary condition for optimality of the problem is given by
the corresponding vakonomic equations. In this case, we will
denote by
\[
y^1=\dot{x}, \ y^2=\dot{y}, \ y^3=\omega_{x},\ y^4=\omega_{y}, \
y^5=\omega_{z}.
\] Therefore, the vakonomic problem is given by the Lagrangian
$$L(t,x,y;y^1,y^2,y^3,y^4,y^5)=\displaystyle\frac{1}{2}\left( (y^1)^2+({y}^2)^2 \right)$$
and the submanifold $\M$ defined by the constraints
\[
\begin{array}{rcl}
y^3&=&\Psi^3(t,{x},
{y},y^1,y^2)=\displaystyle\frac{1}{r}(-{y}^2+\Omega(t)x),\\[8pt]
y^4
&=&\Psi^4(t,{x},y,y^1,y^2)=\displaystyle\frac{1}{r}(y^1+\Omega(t)y),\\[8pt]
y^5&=&\Psi^5(t,{x},y,y^1,y^2)=c.
\end{array}
\]
After some computations, we obtain that the vakonomic equations
are
$$\left\{ \begin{array}{rcl}
\dot{y}_3&=&-\displaystyle\frac{1}{r}({y^1}+\Omega(t)y)y_5+cy_4,\\[8pt]
\dot{y}_4&=&-\displaystyle \frac{1}{r}({y^2}-\Omega(t)x)y_5-cy_3,\\[8pt]
\dot{y}_5&=&
\displaystyle\frac{1}{r}({y^1}+\Omega(t)y)y_3-\frac{1}{r}(-{y^2}+\Omega(t)x)y_4,\\[6pt]
\displaystyle\frac{d}{dt}\left(ry^1-{y_4}\right)&=&-{\Omega}(t)y_3,\\[8pt]
\displaystyle\frac{d}{dt}\left(ry^2+{y_3}\right)&=&-{\Omega}(t)y_4,\\[6pt]
y^1=\dot{x}\, , && y^2=\dot{y}\; .
\end{array}\right .$$
Moreover, it is easy to prove that the vakonomic system is
regular. Therefore, there exists a unique solution of the
vakonomic equations on the submanifold $W'_1$ which is determined
by the following conditions
$$\begin{array}{l}
y_1=\displaystyle\frac{\partial {L}}{\partial y^1}-y_{A}\frac{\partial \Psi^{A}}{\partial y^1}=y^1-\frac{1}{r}y_4,\\
y_2=\displaystyle\frac{\partial L}{\partial
y^2}-y_{A}\frac{\partial
\Psi^{A}}{\partial y^2}=y^2+\frac{1}{r}y_3,\\
y_0={L}-y_A\Psi^A\kern-1pt-y_ay^a\kern-1pt=-\displaystyle\frac{1}{2}(y_1\kern-1pt+\frac{1}{r}y_4)^2\kern-1pt-
\frac{1}{2}(y_2\kern-1pt-\frac{1}{r}y_3)^2\kern-1pt-\frac{\Omega(t)}{r}(xy_3\kern-1pt+yy_4)-cy_5.
\end{array}$$

Thus, it follows that $(t,x,y; y_1, y_2, y_3, y_4, y_5)$
(respectively, $(t,x,y;y_0, y_1, y_2,$ $ y_3, y_4,$ $ y_5)$) are
local coordinates on $W'_1$ (respectively, on $W_1$). Then, the
local expression of the Hamiltonian $H_{W_1}$ is
$$\begin{array}{rcl}
H_{W_1}(t,x,y;y_0,y_1,y_2,y_3,y_4,y_5)&=&-y_0-\displaystyle\frac{1}{2}(y_1+\frac{1}{r}y_4)^2-\frac{1}{2}(y_2-\frac{1}{r}y_3)^2\\[8pt]
&&-\displaystyle \frac{\Omega(t)}{r}(xy_3+yy_4)-cy_5 \end{array}$$
and, in terms of the affine-linear part
$\{\cdot,\cdot\}_{vak}^{al}$ of the vakonomic bracket
$\{\cdot,\cdot\}_{vak}$ associated with the regular system
$(L,\M)$, the vakonomic equations are
$$\left\{ \begin{array}{l}
\dot{y}_1= \{h_1, y_1
\}_{vak}^{al}=-\displaystyle\frac{\Omega(t)}{r}y_3,\;\;\;
\dot{y}_2=\{h_1, y_2
\}_{vak}^{al}=-\displaystyle\frac{\Omega(t)}{r}y_4,\\[8pt]
\dot{y}_3= \{h_1, y_3
\}_{vak}^{al}=-\displaystyle\frac{1}{r}\left(y_1+\displaystyle\frac{1}{r}y_4+{\Omega(t)
y}\right)y_5+cy_4,\\[8pt]
\dot{y}_4= \{h_1,
y_4\}_{vak}^{al}=\displaystyle\frac{1}{r}\left(-y_2+\frac{1}{r}y_3+{\Omega(t)
x}\right)y_5-cy_3,\\[8pt]
\dot{y}_5= \{h_1,y_5\}_{vak}^{al}=\displaystyle\frac{1}{r}\Big[
\left(y_1\kern-0.5pt+\frac{1}{r}y_4+{\Omega(t)y}\right)y_3+\left(y_2-\frac{1}{r}y_3-{\Omega(t){x}}\right)y_4
\Big],\\[8pt]
\dot{x}= \{
h_1,{x}\}_{vak}^{al}=y_1+\displaystyle\frac{1}{r}y_4,\;\;\;
\dot{y}= \{h_1,y\}_{vak}^{al}=y_2-\displaystyle\frac{1}{r}y_3.
\end{array}\right .$$

}\hfill$\triangle$
\end{example}
\subsubsection{Optimal control of affine control systems}\label{subsection-affine}

Let $\tau_{\mathcal A}: {\mathcal A} \to Q$ be a Lie affgebroid.
Suppose that the constraint submanifold $\M$ of the vakonomic
system is an affine subbundle $\B$ of ${\mathcal A}$, that is, we
have an affine bundle $\B$ over $Q$ with associated vector bundle
$\tau_{U_\B}:U_\B\to Q$ and the corresponding inclusions
$i_\B:\B\to\A$ and $i_{U_\B}:U_\B\to V$. Choose now a coordinate
system adapted to this
 affine subbundle $\B$. That is, take local
coordinates $(x^i)$ on an open subset $U$ of $Q$ and an
 local basis of $\Gamma(\tau_{V})$, $\{e_{A}, e_{a}\}$,
adapted to the decomposition $V=\tilde{U}_\B\oplus U_\B$, where
$\tilde{U}_\B$ is an arbitrary  complementary subspace.  Thus,
$\{e_{0}, e_{A}, e_{a}\}$ is an local basis of
$\Gamma(\tau_{\widetilde{\mathcal A}})$ adapted to the affine
subbundle $\B$, where $1_{\mathcal A}(e_0)=1$. Denote by $(x^i,
y^0, y^A, y^a)$ the corresponding local coordinates on
$\widetilde{\mathcal A}$ and by $(x^i, y_0, y_A, y_a)$ the dual
local coordinates on ${\mathcal A}^+$. Note that the equations
which define $\B$ as an affine subbundle of ${\mathcal A}$ are
$y^A = 0$.

The affine control problem given by the  \emph{drift section}
$e_0$ and the  \emph{input sections} $e_a$ is defined by the
following equation on $Q$, $\dot{x}^i=\rho^i_0+y^a \rho^i_a$,
where the coordinates $y^a$ are playing the role of the \emph{set
of admissible controls}.

Now, consider a function $\tilde{L}: \B\longrightarrow \R$ as a
performance index. The equations of motion of the optimal control
problem defined by $(\tilde{L}, \B)$ are precisely the vakonomic
equations. In the selected coordinate system are:
$$\left \{
\begin{array}{l}
\displaystyle \dot{x}^i=\rho_0^i+y^a \rho ^i_a ,\\[4pt]
\displaystyle \dot{y}_{A}= \rho ^i_A\frac{\partial
\tilde{L}}{\partial x^i}  -y_\gamma(C^{\gamma}_{A0}+y^aC^\gamma_{Aa}) ,\\[8pt]
\displaystyle \frac{d}{dt}\Big( \frac{\partial \tilde{L}}{\partial
y^a}\Big)=  \rho ^i_{a}\frac{\partial \tilde{L}}{\partial x^i}
-y_\gamma( C^{\gamma}_{a0}+y^bC^\gamma_{ab}),
\end{array}
\right.$$ for all $1\leq i\leq m$, $1\leq \gamma\leq n$, $1\leq
A\leq\bar{m}$ and $\bar{m}+1\leq a\leq n$, with
$y_a=\frac{\partial \tilde{L}}{\partial y^a}$.

\begin{example}
{\rm

Consider a particle of unit mass in a planar inverse-square law
gravitational field which has thrusters in the ``x, y" directions
(see \cite{GrMa}). Then, the equations of motion are:
$$\dot{q}_1=v_1,\quad \dot{q}_2=v_2,\quad \dot{v}_1=-q_1(q_1^2+q_2^2)^{-3/2}+u_1,\quad
\dot{v}_2=-q_2(q_1^2+q_2^2)^{-3/2}+u_2$$ defined on
$M=(\R^2-\{(0,0)\})\times \R^2$. The objective will be to drive
the particle to a given circular orbit with minimum energy.
Therefore, let us take $L=\frac{1}{2}(u_1^2+u_2^2)$.

Now, choose a global basis of sections of $T(\R\times
M)\longrightarrow \R\times M$ adapted to this affine control
system:
\begin{eqnarray*}
&e_0=\frac{\partial }{\partial t}+v_1\frac{\partial}{\partial q_1}+v_2\frac{\partial}{\partial q_2} -q_1(q_1^2+q_2^2)^{-3/2}\frac{\partial}{\partial v_1} -q_2(q_1^2+q_2^2)^{-3/2}\frac{\partial}{\partial v_2}&\\
&e_1=\frac{\partial}{\partial q_1},\ e_2=\frac{\partial}{\partial
q_2}, e_3=\frac{\partial}{\partial v_1}, \
e_4=\frac{\partial}{\partial v_2}
\end{eqnarray*}
where $\{e_0; e_3, e_4\}$ defines a affine subbundle of $\R\times
TM\longrightarrow \R\times M$ determining the initial affine
control system.

Applying the result developed in Subsection
\ref{subsection-affine} and denoting by   $y_3=u_1$ and $y_4=u_2$
then the  equations of motion are now:
\begin{eqnarray*}
&\dot{q}_1=v_1,\qquad \dot{q}_2=v_2&\\
&\dot{v}_1=-q_1(q_1^2+q_2^2)^{-3/2}+u_1,\qquad \dot{v}_2=-q_2(q_1^2+q_2^2)^{-3/2}+u_2&\\
&\dot{y}_1=-\left(u_1(2q_1^2-q_2^2)+3u_2q_1q_2\right)(q_1^2+q_2^2)^{-5/2}&\\
&\dot{y}_2=-\left(3u_3q_1q_2+u_4(2q_2^2-q_1^2)\right)(q_1^2+q_2^2)^{-5/2}&\\
&\dot{u_1}=-y_1\qquad \dot{u_2}=-y_2 .&
\end{eqnarray*}
}\hfill$\triangle$

\end{example}

\section{Conclusions and future work}
Variational constrained Mechanics is discussed in the Lie
affgebroid setting. We obtain the vakonomic equations and the
vakonomic bracket associated with a constrained mechanical system
on a Lie affgebroid. The variational character of the theory is
analyzed. Vakonomic systems subjected to affine constraints are of
special interest. Other examples are also discussed.

In this paper we only consider normal solutions of the vakonomic
problems. It would be interesting to extend the results of the
paper for abnormal solutions.

\section*{Acknowledgments}
This work has been partially supported by MEC (Spain) Grants MTM
2006-03322, MTM 2007-62478, project ``Ingenio Mathematica"
(i-MATH) No. CSD 2006-00032 (Consolider-Ingenio 2010) and
S-0505/ESP/0158 of the CAM.

\end{document}